\pdfoutput=1
\RequirePackage{ifpdf}
\ifpdf 
\documentclass[pdftex]{sigma}
\else
\documentclass{sigma}
\fi

\usepackage{tikz-cd}

\usepackage{mathtools}

\numberwithin{equation}{section}
\newtheorem{Theorem}{Theorem}[section]
\newtheorem{Lemma}[Theorem]{Lemma}
\newtheorem{Proposition}[Theorem]{Proposition}
\newtheorem{Corollary}[Theorem]{Corollary}

\theoremstyle{definition}
\newtheorem{Definition}[Theorem]{Definition}
\newtheorem{Remark}[Theorem]{Remark}

\makeatletter
\newcommand{\hathat}[1]{%
\begingroup%
 \let\macc@kerna\z@%
 \let\macc@kernb\z@%
 \let\macc@nucleus\@empty%
 \hat{\raisebox{.2ex}{\vphantom{\ensuremath{#1}}}\smash{\hat{#1}}}%
\endgroup%
}
\makeatother

\newcommand{\Z}{\mathbb{Z}}
\newcommand{\CC}{\mathbb{C}}

\begin{document}
\allowdisplaybreaks

\newcommand{\arXivNumber}{2206.10578}

\renewcommand{\PaperNumber}{026}

\FirstPageHeading

\ShortArticleName{On Generalized WKB Expansion of Monodromy Generating Function}

\ArticleName{On Generalized WKB Expansion\\ of Monodromy Generating Function}

\Author{Roman KLIMOV}

\AuthorNameForHeading{R.~Klimov}

\Address{Department of Mathematics and Statistics, Concordia University,\\
1455 de Maisonneuve W., Montreal, QC H3G 1M8, Canada}
\Email{\href{mailto:roman.klimov@concordia.ca}{roman.klimov@concordia.ca}}

\ArticleDates{Received June 22, 2022, in final form April 11, 2023; Published online April 28, 2023}

\Abstract{We study symplectic properties of the monodromy map of the Schr\"odinger equation on a Riemann surface with a meromorphic potential having second-order poles. At first, we discuss the conditions for the base projective connection, which induces its own set of Darboux homological coordinates, to imply the Goldman Poisson structure on the character variety. Using this result, we extend the paper [\textit{Theoret. and Math. Phys.} \textbf{206} (2021), 258--295, arXiv:1910.07140], by performing generalized WKB expansion of the generating function of monodromy symplectomorphism (the Yang--Yang function) and computing its first three terms.}


\Classification{53D30; 34M45; 34E20}

\section{Introduction}
The symplectic aspects of the monodromy map of Schr\"odinger and related Schwarzian equations is a rich topic that draws attention of contributors from mathematical and physical communities. This study was initiated by S.~Kawai \cite{Kawai} who established a relationship between the canonical symplectic structure on the cotangent bundle $T^* \mathcal{M}_{g}$ of the moduli space of curves and Goldman's bracket for the traces of monodromy matrices. Kawai's results have found a physics application in the geometry of four-dimensional supersymmetric quantum field theories in \cite{nekrasov, tesch}.

Later in \cite{Bertola_2017, korotkin2018periods} authors proposed an alternative approach to the symplectic geometry of the monodromy map using homological coordinates to code information about a potential and complex structure of a Riemann surface, in case when potential is holomorphic or with first-order poles. These works highly relied on the canonical identification between associated moduli spaces of quadratic differentials and cotangent bundles $T^* \mathcal{M}_{g,n}$ of moduli spaces of punctured surfaces. In this paper, we will generalize their results by considering quadratic differentials with second-order poles. Although in terms of monodromy data this case is well-studied thanks to Fuchsian nature of the singularities, from the symplectic point of view it brings technical difficulties due to absence of the aforesaid identification of moduli spaces and requires a different approach that we develop. Further generalization to include poles of order three and more poses an issue on account of emergence of generalized monodromy data, including Stokes and connection matrices~\cite{Birk}. That would require a non-obvious choice of local Darboux coordinates on the space of parameters of the equation. Recent developments in this direction for genus~0 are present in~\cite{harnad}. For a higher genus this remains a challenging problem.

Introduce the linear second-order equation on a Riemann surface $\mathcal{C}$ of genus $g$ with $n$ punctures in the form
\begin{equation}\label{1}\partial^2 \phi+U\phi=0,\end{equation}
where $U$ is a meromorphic potential on $\mathcal{C}$ with double poles at the punctures $(z_j)^n_{j=1}$. Invariance of the equation under a coordinate change implies that $U$ transforms as a projective connection, while the
solution $\phi$ locally transforms as $\frac{1}{2}$-differential~\cite{halfdif}. To parametrize the space of all meromorphic potentials, we represent the potential $U$ as
\begin{equation}\label{pot}U=\frac{1}{2}S-Q,\end{equation}
where $S$ is a fixed projective connection on $\mathcal{C}$ with at most simple poles at $z_j$ and holomorphically depending on moduli of $\mathcal{M}_{g,n}$, while the quadratic differential $Q$, having double poles at $z_j$, varies. Let the asymptotics of $Q$ near the poles be given by
\begin{equation*}
Q(x) \sim \left(\frac{r_j^2}{\xi_j^2} +O\big(\xi_j^{-1}\big)\right)({\rm d} \xi_j)^2.
\end{equation*}

Denote by $\mathcal{Q}_{g,n}$ the moduli space of pairs $(\mathcal{C},Q)$ and by $\mathcal{Q}_{g,n}[\mathbf{r}]$ its corresponding stratum for fixed values of $r_j$'s. The solution to \eqref{1} is locally a $-\frac{1}{2}$ differential which could be written as
\begin{equation*}
\phi=\phi(\xi)({\rm d}\xi)^{-\frac{1}{2}}.
\end{equation*}
The ratio $f={\phi_1}/{\phi_2}$ of two
linearly independent solutions of \eqref{1} solves the Schwarzian equation
\begin{equation*}
\{f,\xi\}=S(\xi)-2Q(\xi),
\end{equation*}
where $\xi$ is an arbitrary local parameter on $\mathcal{C}$ and $\{ \, \cdot \, , \, \cdot \, \}$ denotes the Schwarzian derivative. Analytic continuation of $f$ along the cycles of $\pi(\mathcal{C}\backslash \{ z_i\}^n_{j=1}, x_0)$ determines a ${\rm PSL}(2, \CC)$ monodromy representation of the fundamental group with the chosen basepoint $x_0$. The choice of standard generators $\big(\{\kappa\}^{n}_{j=1},\{\alpha, \beta\}^{g}_{j=1}\big)$ of the fundamental group with single relation
\begin{equation*}
\kappa_1 \cdots \kappa_n \prod^g_{i=1}\alpha_i \beta_i \alpha_i^{-1} \beta_i^{-1}={\rm id}
\end{equation*} yields the same relation on the monodromy matrices
\begin{equation}\label{5}
M_{\kappa_1}\cdots M_{\kappa_n}\prod^g_{i=1}M_{\alpha_i} M_{\beta_i} M_{\alpha_i}^{-1}M_{\beta_i}^{-1}=I.
\end{equation}
The matrix $M_{\kappa_j}$ corresponding to the monodromy around the pole $z_j$ has the following diagonal form:
\begin{equation*}
D_j=\begin{pmatrix}m_j & 0 \\ 0 & m_j^{-1} \end{pmatrix},
\end{equation*}
where
\begin{equation}\label{7}
 m_j^2={\rm e}^{4 \pi {\rm i} \lambda_j}.
\end{equation}
We additionally assume that $\lambda_j \notin \Z/2$ to exclude the resonant case.
Local analysis of the solutions for \eqref{1} implies the following relation between the biresidues $(r_j)$ and eigenvalues~$(m_j)$%
\begin{equation}\label{8}{r^2_j}=\lambda_j(\lambda_j-1). \end{equation}
We denote by ${\rm CV}_{g,n}$ the ${\rm PSL}(2)$ character variety corresponding to the representation~\eqref{5}. It is well known that the stratum ${\rm CV}_{g,n}[\mathbf{m}]$ for fixed values $m_j$ is a symplectic leaf with a Poisson structure given by the Goldman bracket~\cite{Goldman_1984}.

The space $\mathcal{Q}_{g,n}[\mathbf{r}]$ admits a system of local coordinates defined in the following way: for every pair $( \mathcal{C}, Q) \in \mathcal{Q}_{g,n} $ consider equation $v^2=Q$ in $T^*\mathcal{C}.$ Although $\sqrt{Q}$ is not single valued on $\mathcal{C},$ the equation induces a branched double covering $\pi\colon \hat{\mathcal{C}} \xrightarrow{} \mathcal{C},$ where $v$ is a single-valued Abelian differential. The map $\pi$ is branched at (simple) zeroes of $Q$ denoted by $(x_j)^{4g-4+2n}_{j=1}$. Thus, each double pole $z_j$ has two preimages that we call $\big(z_j^{(1)}, z_j^{(2)}\big)$. The enumeration of these points is chosen
such that the residue of $v$ at $z_j^{(1)}$
equals $r_j$ and the residue of $v$ at $z_j^{(2)}$
equals~$-r_j$. The genus of $\hat{\mathcal{C}}$ is $\hat{g}=g+g^-$, with $g^-=3g-3+n.$ The surface $\hat{\mathcal{C}}$ is equipped with a natural holomorphic involution $\mu\colon \hat{\mathcal{C}} \xrightarrow{} \hat{\mathcal{C}}$ which interchanges the sheets of the double cover. The involution induces the splitting of the homology group $H_1\big(\hat{\mathcal{C}}\backslash (z_j^{(1)}, z_j^{(2)})^{n}_{j=1}, \Z\big)$ into even $H_+$ and odd $H_-$ parts. We choose appropriate subset of cycles $(a^-_i,b^-_i)^{g^-}_{i=1} \in H_-$ with intersection index $a^-_i \circ b^-_j=\frac{1}{2}\delta_{ij}$ so that the integrals
\begin{equation*}
A_j=\oint_{a^-_j}v, \qquad B_j=\oint_{b^-_j}v
\end{equation*}
become local period (or homological) coordinates on $\mathcal{Q}_{g,n}[\mathbf{r}]$ \cite{Kontsevich_2003}. The intersection pairing defines the natural symplectic form on $\mathcal{Q}_{g,n}[\mathbf{r}]$
\begin{equation*}
\Omega_{{\rm hom}}=\sum^{g^-}_{j=1}2 \delta B_j \wedge \delta A_j.
\end{equation*}
\begin{Remark}

Throughout this article we will use the ``$\delta$'' symbol as differential with respect to moduli, while ``${\rm d}$'' refers to the differential with respect to some local coordinate near a point on surface.
\end{Remark}
Here and below we will assume that $Q$ is free from saddle trajectories (i.e., it is a ``Gaiotto--Moore--Neitzke differential'' \cite{giatto}), so that the symplectic form on ${\rm CV}_{g,n}[\mathbf{m}]$ that inverts the Goldman bracket could be written in terms of homological shear coordinates given by linear combinations of the logarithms of classical Thurston's shear coordinates \cite{thurston} emerging from the ideal triangulation of the Riemann surface $\mathcal{C}$ (they are a simple example of more involved Fock--Goncharov coordinates \cite{fock}):
\begin{equation}\label{11}\Omega_{G}=\sum^{g^-}_{j=1} 2 \delta\rho_{a^-_j} \wedge \delta\rho_{b^-_j}.\end{equation}

Introduce the Bergman projective connection $S_B$ defined in terms of the canonical bidifferential $B(x,y)$ on $\mathcal{C}$, which is normalized with respect to chosen Torelli marking in $H_1(\mathcal{C},\Z)$:%
\begin{equation}\label{12}
 S_B(x)=\left(B(x,y)-\frac{{\rm d}\xi(x) {\rm d}\xi(y)}{(\xi(x)-\xi(y))^2}\right)\Big|_{y=x},
\end{equation}
where $\xi$ is any local coordinate near point $x.$ Since $S_B$ depends holomorphically on the conformal structure of $\mathcal{C},$ the difference $S-S_B$ becomes
a family of quadratic differentials with at most simple poles at the punctures $(z_j)$, depending holomorphically on moduli of $\mathcal{M}_{g,n}.$ Using the identification of the moduli space of quadratic differentials with simple poles and the cotangent bundle $T^*\mathcal{M}_{g,n}$, we can associate $S-S_B$ with the 1-form $\Theta_{(S-S_B)}$, locally defined on $\mathcal{M}_{g,n}$, in the following way.

At first, introduce the set of holomorphic local coordinates $(\Omega_{jk}, q_l)$ on $\mathcal{M}_{g,n}, \, g \geq 2$. To determine locally the
conformal structure of $\mathcal{C}$ we pick at generic point of $\mathcal{M}_{g,n}$ (outside of hyperelliptic locus for $g \geq 3$) a set $D$ of $3g-3$ entries of the
period matrix $\Omega$ of $\mathcal{C}.$ The quadratic differentials corresponding to cotangent vectors $\delta \Omega_{jk} $ are products $u_j u_k$ of normalized
holomorphic differentials. An additional set of $n$ coordinates which determine the positions of punctures $(z_l)^n_{l=1} $
on $\mathcal{C}$ we choose to be $q_l = (u_i/u_j )(z_l)$ where $u_i$ and $u_j$ form a pair of normalized holomorphic 1-forms on
$\mathcal{C}$, such that $u_j(z_l) \neq 0$. The quadratic differential corresponding to cotangent vector $\delta q_l$ is the meromorphic quadratic
differential $Q^{z_l}$ (given by the formula \eqref{110} below) whose only simple pole is at $z_l$. These coordinates are local: in different coordinate charts on
$\mathcal{M}_{g,n}$ one might need to choose other pairs of normalized holomorphic differentials and/or different Torelli
markings. The momenta $p_l$ are then defined
to be coefficients of decomposition of the quadratic differential $S-S_B$ in the
basis described above.
Writing down the quadratic differential $S-S_B$ as
\begin{equation*}
S-S_B=\sum_{(jk) \in D}p_{jk} u_j u_k+ \sum^n_{l=1} p_l Q^{z_l},
\end{equation*}
where $p_{jk}$ and $p_l$ are holomorphic functions of $(\Omega_{jk},q_l)$, the corresponding 1-form $\Theta_{(S-S_B)}$ on~$\mathcal{M}_{g,n}$ reads as
\begin{equation}\label{form}\Theta_{(S-S_B)}=\sum_{(jk) \in D}p_{jk} \delta \Omega_{jk}+ \sum^n_{l=1} p_l \delta q_l.\end{equation}
Local coordinates on $\mathcal{M}_{g,n}$ for $ g=0, 1$ have a special description and were covered in~\cite{korotkin2018periods}.

Our first main result imposes a condition on projective connection $S$ of \eqref{pot} for the monodromy map to become a symplectomorphism.

\begin{Theorem} \label{Th_1}
The monodromy map
\begin{equation*}
\mathcal{F}_{(S)}\colon \ \mathcal{Q}_{g,n}[\mathbf{r}] \xrightarrow[]{} {\rm CV}_{g,n}[\mathbf{m}]
\end{equation*} is a symplectomorphism with $\mathcal{F}_{(S)}^*\Omega_G=-\Omega_{\hom}$
if and only if the $1$-form $\Theta_{(S-S_B)}$, corresponding to family of quadratic differentials $S-S_B$ $($which is locally
defined on the moduli space $\mathcal{M}_{g,n})$,
is closed, $\delta\Theta_{(S-S_B)}=0.$
\end{Theorem}

Statement of the theorem generalizes the results proven in \cite{Bertola_2017, korotkin2018periods}, where the differential~$Q$ is assumed to be holomorphic or with simple poles, respectively. The proofs were based on the identification of the homological symplectic form with the canonical form on $T^*\mathcal{M}_{g,n}$ which does not hold in presence of second-order poles. Our proof involves a perturbation of quadratic differential~$Q$ and expansion of homological coordinates by series. A~similar approach was employed before to study Strebel differentials~\cite{DuaHub}. The outlined criterion effectively proves a symplectic nature of the monodromy map for a large class of projective connections which are known to satisfy the given condition of the closedness of 1-form (for example, Schottky, Wirtinger and Bers projective connections~\cite{Bertola_2017}).

Let us take $S=S_B$. Choosing symplectic potentials on $\mathcal{Q}_{g,n}[\mathbf{r}]$:
\begin{equation*}
\theta_{{\rm hom}}=\sum^{g^-}_{j=1} \big(B_j \delta A_j - A_j \delta B_j \big)\end{equation*}
and on ${\rm CV}_{g,n}[\mathbf{m}]$:
\begin{equation*}
\theta_{G}=\sum^{g^-}_{j=1}\big( \rho_{b^-_j} \delta \rho_{a^-_j} - \rho_{a^-_j} \delta \rho_{b^-_j} \big)
\end{equation*}
we may consider the generating function of this symplectomorphism (which was called the Yang--Yang function in \cite{bertola2021wkb} following the seminal work~\cite{nekrasov}) given by
 \begin{equation*}
 \delta \mathcal{G}_B=\mathcal{F}_{(S_B)}^*\theta_{G}-\theta_{{\rm hom}}.
 \end{equation*}
 Our definition of $\mathcal{G}_B$ involves a different choice of the Darboux coordinates on the character variety (homological shear coordinates versus complex Fenchel--Nielsen coordinates). Thus, the actual Yang--Yang function defined in \cite{nekrasov} differs from $\mathcal{G}_B$ by a generating function of the change of these coordinates.\footnote{Another difference with the (Nekrasov--Rosly--Shatashvili) Yang--Yang function of \cite{nekrasov} is that $\mathcal{G}_B$ is explicitly
dependent on chosen projective structure on the base surface as well as the moduli of the spectral cover,
whereas the NRS Yang--Yang function is a function of the complex structure parameters of the base surface.}
While $\mathcal{G}_B$ is invariant under symplectic transformations of the basis $(a^-_i,b^-_i) \in H_{-} $, its dependance on the Torelli marking of $\mathcal{C}$ was unclear in~\cite{bertola2021wkb}. The following result addresses this issue.
\begin{Proposition}
Let two Torelli markings $\alpha^\sigma$ and $\alpha$ be related by ${\rm Sp}(2g,\Z)$ matrix
\begin{equation*}
\sigma=\begin{pmatrix}
C & A \\
D & B
\end{pmatrix}\colon \qquad
\begin{pmatrix}
b \\
a
\end{pmatrix}^{\sigma}= \sigma
\begin{pmatrix}
b \\
a
\end{pmatrix}.
\end{equation*}
Under this change the monodromy generating function transforms as
\begin{equation*}
\mathcal{G}^{\sigma}_B=\mathcal{G}_B+\sum^n_{i=1}\pi {\rm i} r_i\left({\rm reg}\int^{z^{(1)}_i}_{z^{(2)}_i}v_1 -{\rm reg}\int^{z^{(1)}_i}_{z^{(2)}_i}v_0 \right) +6 \pi {\rm i} \log \det (C \Omega +D),
\end{equation*}
where
\begin{equation*}
v_0^2=Q, \qquad v_1^2=Q+6 \pi {\rm i} \sum_{1 \leq j \leq k \leq g}u_j u_k \frac{\partial}{\partial \Omega_{jk}}\log \det (C \Omega +D)
\end{equation*}
 define two canonical coverings. The regularization of the integrals at the endpoints is given by~\eqref{76}.
\end{Proposition}

Although $\mathcal{G}_B$ is defined very implicitly, the WKB approximation of the equation~\eqref{1} allows us to compute its asymptotic expansion. In~\cite{bertola2021wkb}, authors studied the following equation for a~small parameter $\hbar \ll 1$
\begin{equation*}\partial^2 \phi+\left(\frac{1}{2}S_B-\frac{Q}{\hbar^2}\right)\phi=0\end{equation*}
and, following \cite{voros_alegr, AlBrid}, established a link between the $\hbar$-expansion of the homological shear coordinates \eqref{11} and Voros symbols -- integrals over the odd homology group $H_-$ \cite{Voros}. That allowed them to compute the leading asymptotic of the WKB expansion of the generating function $\mathcal{G}_B$. In this paper, we use the result of Theorem \ref{Th_1} to generalize the WKB method by considering Schr\"odinger equation with the potential perturbed by the $\frac{1}{\hbar}$-term: let $Q_1$ be a meromorphic differential on $\mathcal{C}$ with at most simple poles at the punctures $(z_j)^n_{j=1}$, holomorphically depending on moduli of $\mathcal{M}_{g,n}$. Take the base projective connection
\begin{equation*}
S=S_B-\frac{2 Q_1}{\hbar} \end{equation*}
and consider second-order equation \eqref{1} on a Riemann surface $\mathcal{C}$ in the form
\begin{equation}\label{19}
\partial^2 \phi+\left(\frac{1}{2}S_B-\frac{{Q}_1}{\hbar}-\frac{Q}{\hbar^2}\right)\phi=0.
\end{equation}

Such equations in genus zero are known to give rise to Riemann--Hilbert problems emerging naturally in Donaldson--Thomas theory \cite{Bri, BM} and appear in the quantization of spectral curves via the topological recursion~\cite{iwaki}. Our result involves a computation of three leading terms of the $\hbar$-expansion of the generating function $\mathcal{G}_B$ for this equation on a~Riemann surface of arbitrary genus.

\begin{Theorem}
Consider the monodromy map
\begin{equation*}\mathcal{F}_{(S_B-2Q_1/\hbar)}\colon \ \mathcal{Q}_{g,n}[\mathbf{r}/\hbar] \xrightarrow[]{} {\rm CV}_{g,n}[\mathbf{m}(\hbar)]\end{equation*}
of equation \eqref{19} between the moduli space $\mathcal{Q}_{g,n}[\mathbf{r}/\hbar]$ of pairs $\big(\mathcal{C}, Q/\hbar^2\big)$ and the symplectic leaf ${\rm CV}_{g,n}[\mathbf{m}(\hbar)]$ of the ${\rm PSL}(2, \CC)$ character variety, where each $m_j(\hbar)$ is related to $r_j$ via $\eqref{7},\eqref{8}$ by
\begin{equation*}
\frac{r_j}{\hbar}=\pm \left[\frac{\log m_j}{2 \pi {\rm i} } \left(\frac{\log m_j}{2 \pi {\rm i} }-1 \right)\right]^{1/2}. \end{equation*}
Then the map $\mathcal{F}_{(S_B-2Q_1/\hbar)}$ is a symplectomorphism if and only if the 1-form $\Theta_{(Q_1)}$, which is locally defined on $\mathcal{M}_{g,n}$ by \eqref{form}, is closed, $\delta\Theta_{(Q_1)}=0.$ The monodromy generating function $\mathcal{G}_B$ has the following asymptotics as $\hbar \xrightarrow{} 0^+$:
\begin{equation*}
\mathcal{G}_B(\hbar)=\frac{G_{-1}}{\hbar}+G_0+G_1 \hbar + O\big(\hbar^2\big).\end{equation*}
Here
\begin{equation*}
G_{-1} =\Hat{G}_{(Q_1)}+\sum^n_{j=1}\frac{\pi {\rm i} r_j}{2}
\int^{z^{(1)}_j}_{z^{(2)}_j}\frac{{Q}_1}{v} , \end{equation*} where there exists a local holomorphic function $\Hat{G}_{(Q_1)}$ on $\mathcal{M}_{g,n}$, such that
\begin{equation*}\delta\Hat{G}_{(Q_1)}=\Theta_{(Q_1)}.\end{equation*}
 Explicit form of $\Hat{G}_{(Q_1)}$ depends on the concrete choice of $Q_1$;
\begin{equation*}
G_{0}=-12 \pi {\rm i} \log \tau_B|_{r} -\sum^n_{j=1}\frac{\pi {\rm i} r_j}{2} \int^{z^{(1)}_j}_{z^{(2)}_j}\left(qv+\frac{1}{4r^2_j}v \right)+\sum^n_{j=1}
\pi {\rm i} r_j \binom{\frac{1}{2}}{2}\int^{z^{(1)}_j}_{z^{(2)}_j}\frac{{Q^2_1}}{v^{3}},
\end{equation*}
here\footnote{Here and futher the binomial coefficient is defined by \textit{$\binom{\frac{1}{2}}{k}:=\frac{1/2(1/2-1)\cdots (1/2-k+1)}{k (k-1) \cdots 1}$}.}
$\log \tau_B|_{r}$ is the Bergman tau-function defined on stratum of the moduli space of quadratic differentials with $n$ second-order poles with fixed biresidues $($see~{\rm \cite{Bertola_2020}} for the definition and main properties. In fact, $\tau_B$ could be viewed as a section of the determinant line bundle of the Hodge vector
bundle over the moduli space $\mathcal{Q}_{g,n}[\mathbf{r}])$, $q(x)$ is a meromorphic function on $\mathcal{C}$ given by
\begin{equation}\label{47}q(x)=\frac{S_B-S_v}{2v^2}, \end{equation}
where $S_v$ is the Schwarzian projective connection
\begin{equation*}
S_v(\xi(x))=\bigg\{\int^x v, \xi(x) \bigg\}({\rm d} \xi(x))^2,
 \end{equation*}
$(\{\cdot, \xi\}$ is the Schwarzian derivative, $\xi$ is a local coordinate on $\mathcal{C})$;
\begin{gather*}
G_1=-\sum^{4g-4+2n}_{i=1} \frac{5 \pi {\rm i}}{72}\underset{x_i}{\rm res}\left(\frac{Q_1/v}{\int^x_{x_i}v} \right)+\sum^{n}_{j=1}\frac{\pi {\rm i} r_j }{4}\int^{z^{(1)}_j}_{z^{(2)}_j} q\frac{Q_1}{v} \\
\hphantom{G_1=}{}
+\sum^{n}_{j=1}\frac{\pi {\rm i}}{16{r}_j}\int^{{z}^{(1)}_j}_{{z}^{(2)}_j}\frac{Q_1}{v}+\sum^n_{j=1}\pi {\rm i} r_j\binom{\frac{1}{2}}{3}
 \int^{z^{(1)}_j}_{z^{(2)}_j}\frac{Q^3_1}{v^{5}},
 \end{gather*}
where the first sum runs over the branch points $(x_i)$ of the double cover.
\end{Theorem}

This paper is organized as follows: in Section~\ref{section2}, we describe the geometry and main objects associated with the canonical double cover, Section~\ref{section3} is devoted to the symplectic properties of the monodromy map. In Section~\ref{section4}, we perform the generalized WKB expansion of the monodromy generating function and compute its first few terms.

\section{Spaces of quadratic differentials with second-order poles}\label{section2}
\subsection{Canonical double cover}

Denote by $\mathcal{Q}_{g,n}$ the moduli space of meromorphic quadratic differentials on Riemann surface $\mathcal{C}$ of genus $g$ with $n$ double poles $(z_1, \dots, z_n)$ and $4g-4+2n$ simple zeroes $(x_1,\dots,x_{4g-4+2n})$. We assume that any quadratic differential $Q \in \mathcal{Q}_{g,n}$ has the following asymptotics near poles:
\begin{gather*}
Q(x) \sim \left(\frac{r_j^2}{\xi_j^2} +O(\xi_j^{-1})\right)(d \xi_j)^2,
\end{gather*}
as $x \xrightarrow{} z_i,$ here $\xi_j$ is any local coordinate near pole $z_j$.
For all such $Q$ the equation $v^2=Q$ in the cotangent bundle $T^*\mathcal{C}$ defines double covering
$\pi\colon \hat{\mathcal{C}}\xrightarrow{}\mathcal{C},$ branched at zeroes of $Q.$ The covering surface $\hat{\mathcal{C}}$ possesses a natural holomorphic involution $\mu\colon \hat{\mathcal{C}} \xrightarrow{}\hat{\mathcal{C}}.$ The differential $v$ is single-valued and meromorphic on $\hat{\mathcal{C}}$, skew-symmetric under the involution: $v(x^{\mu})=-v(x)$. $v$ has double zeroes at branch points $(x_j)^{4g-4+2n}_{j=1}$ and simple poles at $2n$ preimages of $(z_j)^{n}_{j=1}$ denoted by $z^{(1)}_j$ and $z^{(2)}_j$ with residues $r_j$ and $-r_j$, respectively. The Riemann--Hurwitz formula implies the genus of the covering surface $\hat{\mathcal{C}}$ equals
\begin{equation*}
\hat{g}=4g-3+n.
\end{equation*}
We decompose the first homology group of $H_1\big(\hat{\mathcal{C}}\backslash \big\{z^{(1)}_j, z^{(2)}_j \big\}^n_{j=1}, \Z \big)$ into
 \begin{equation*}
 H_1\big(\hat{\mathcal{C}}\backslash \big\{z^{(1)}_j, z^{(2)}_j \big\}^n_{j=1}, \Z \big)=H_+ \oplus H_- ,
 \end{equation*}
which are the $+1$ and $-1$ eigenspaces of the map, induced by the involution $\mu.$ $\dim(H_+)=2g+n-1$ and $\dim(H_-)=6g-6+3n:=2g^-+n $.
The canonical basis of $H_1\big(\hat{\mathcal{C}}\backslash \big\{z^{(1)}_j, z^{(2)}_j \big\}^n_{j=1}, \Z \big)$ can be chosen as follows:
\begin{equation*}
\big\{a_k, a^\mu_k, \tilde{a}_l, b_k, b^\mu_k, \tilde{b}_l, t_j, t^\mu_j \big\}, \qquad k=1,\dots,g, \quad l=1,\dots,2g-3+n, \quad j=1,\dots,n.
\end{equation*}
Here
$\big\{ a_k, b_k, a^\mu_k, b^\mu_k \big\} $ is a lift of the canonical basis of cycles $\{ a_k ,b_k \}$ from $\mathcal{C}$ to $\hat{\mathcal{C}}$ such that
\begin{equation*}
\mu_*a_k=a^\mu_k, \qquad \mu_*b_k=b^\mu_k.
\end{equation*}
The cycles $\big\{\tilde{a}_l, \tilde{b}_l\big\}$ are skew-symmetric double covers of the branch cuts between corresponding pairs of zeroes of differential $Q$ with the condition
\begin{equation*}
 \mu_*\tilde{a}_l+\tilde{a}_l=\mu_*\tilde{b}_l+\tilde{b}_l=0. \end{equation*}
Cycles $\big\{t_j, t_j^\mu \big\}$ are two lifts of a small positively oriented loop~$t_j$ around~$z_j$ on~$\mathcal{C}$. On double cover~$\hat{\mathcal{C}}$, $t_j$~denotes a positively oriented loop encircling $z^{(1)}_j$, while $t_j^\mu$ is a small loop around~$z^{(2)}_j$. In the group $H_{+}$ there is a single relation given by
\begin{equation*}
\sum^{n}_{j=1}\big(t_j+t^\mu_j\big)=0. \end{equation*}
The classes
\begin{equation*}
a^+_k=\frac{1}{2}\big(a_k+a^\mu_k\big), \qquad b^+_k=\frac{1}{2}\big(b_k+b^\mu_k\big), \qquad t^+_j=\frac{1}{2}\big(t_j+t^\mu_j\big)\end{equation*}
 generate the group $H_+$ with the intersection index
\begin{equation*}
a^+_i \circ b^+_k=\frac{1}{2}\delta_{ik},
\end{equation*}
while $t^+_j$'s have zero intersection with all cycles.
The following cycles
\begin{gather}\label{32}a^-_k=\frac{1}{2}\big(a_k-a^\mu_k\big), \qquad b^-_k=\frac{1}{2}\big(b_k-b^\mu_k\big),\\
\label{33}a^-_l=\frac{1}{\sqrt{2}}\tilde{a}_l, \qquad b^-_l=\frac{1}{\sqrt{2}}\tilde{b}_l,\\
t^-_j=\frac{1}{2}\big(t_j-t^\mu_j\big)\nonumber
\end{gather}
are the generators of the group $H_-$. Similarly, their intersection index is
\begin{equation*}
a^-_i \circ b^-_k=\frac{1}{2}\delta_{ik}
\end{equation*}
and all other intersections are zero.

The differential $v$ is used to introduce
a system of local coordinates on both $\mathcal{C}$ and $\hat{\mathcal{C}}.$ If $x$ is a~point of $\hat{\mathcal{C}}$ which does not coincide with branch points $\{x_i\}$ and poles $\big\{ z^{(1)}_i, z^{(2)}_i \big\}$ then the local coordinate (also called ``flat'' coordinate) near $x$ is given by
\begin{equation}\label{36}
z(x)=\int^x_{x_1}v,
\end{equation}
$x_1$ is a chosen ``first'' zero of $v.$ $z(x)$ could also be used as a coordinate on $\mathcal{C}$ outside branch points and poles. Notice that in this case $v=dz.$ Near a branch point $x_i$ on $\hat{\mathcal{C}}$ the \textit{distinguished} local coordinate is given by
\begin{equation}\label{37}
\hat{\xi}_i(x)=\left(\int^x_{x_i}v\right)^{\frac{1}{3}}.
\end{equation}
On the curve $\mathcal{C}$ the local coordinate near $x_i$ is
\begin{equation}\label{38}{\xi}_i(x)=\hat{\xi}^2_i(x)=\left(\int^x_{x_i}v\right)^{\frac{2}{3}}. \end{equation}
Near a double pole $z_i$ on $\mathcal{C}$ and simple poles $\big(z^{(1)}_i, z^{(2)}_i\big)$ on $\hat{\mathcal{C}}$ the local coordinate is
\begin{equation}\label{39}\zeta_i(x)=\exp \left(\frac{1}{ r_i} \int^x_{x_1}v\right). \end{equation}
To define local coordinates near the poles uniquely, we on $\mathcal{C}$ connect first zero $x_1$ with a chosen first double pole~$z_1$ by a branch cut, then connect $z_1$ with the remaining poles $\{z_i\}^n_{i=2}$ forming a tree. Then we lift this tree to $\hat{\mathcal{C}}$ via $\pi^{-1}.$

\subsection{Period coordinates. Homological symplectic form}
The dimension of $H_-$ coincides with the dimension of $\mathcal{Q}_{g,n}$. We introduce the following set of \textit{period} (homological) local coordinates on $\mathcal{Q}_{g,n}$:
\begin{equation*}
A_j=\oint_{a^-_j}v, \qquad B_j=\oint_{b^-_j}v, \qquad 2 \pi r_k= \oint_{t^-_k}v. \end{equation*}
We fix the values $(r_k)$ and denote the corresponding stratum of moduli space by $\mathcal{Q}_{g,n}[\mathbf{r}]$. Then $(A_j, B_j)^{g^-}_{j=1}$ become local coordinates on $\mathcal{Q}_{g,n}[\mathbf{r}].$
There is a natural Poisson structure on $\mathcal{Q}_{g,n}$ between periods of $v$ induced by the intersection index of the corresponding cycles $s_1, s_2 \in H_{-}$:
\begin{equation}\label{41}\bigg\{\int_{s_1}v, \int_{s_2}v \bigg\}= s_2 \circ s_1 .\end{equation}
This Poisson structure is degenerate, with Casimir functions $r_1,\dots,r_n.$ The stratum $\mathcal{Q}_{g,n}[\mathbf{r}]$ becomes a symplectic leaf for given bracket.
It allows us to introduce the following symplectic form on $\mathcal{Q}_{g,n}[\mathbf{r}]$:
\begin{equation*}
\Omega_{{\rm hom}}=\sum^{g^-}_{j=1}2 \delta B_j \wedge \delta A_j.
\end{equation*}

\subsection{Variational formulas}
In this section, we introduce several meromorphic functions associated with surfaces~$\hat{\mathcal{C}}$ and~$\mathcal{C}.$ Then we recall their variational formulas with respect to period coordinates on the moduli space~$\mathcal{Q}_{g,n}$. Let us denote by $(u_j)^{g}_{j=1} \in H^{1,0}({\mathcal{C}})$ the basis of holomorphic differentials normalized over $a$-cycles of $H_1({\mathcal{C}})$.
\begin{itemize}\itemsep=0pt
\item The matrix $\Omega_{ij}=\oint_{b_j}u_i$ represents the $g \times g$ period matrix of the base surface $\mathcal{C}$.

\item Meromorphic functions $f_j\colon \hat{\mathcal{C}} \xrightarrow{} \CC P^1$ given by
\begin{equation*}
f_j(x)=\frac{u_j(x)}{v(x)}, \qquad j=1,\dots,g.
\end{equation*}
These functions are skew-symmetric under the involution and generically (when zeroes of~$u_j(x)$ and~$v(x)$ differ) have simple poles at the branch points $(x_j)^{4g-4+2n}_{j=1}$.

\item The meromorphic function $q\colon {\mathcal{C}} \xrightarrow{} \CC P^1 $
\begin{equation} \label {q_form} q(x)=\frac{S_B-S_v}{2v^2}, \end{equation}
where $S_v$ is the Schwarzian projective connection defined by
\[
S_v(\xi(x))=\left\{\int^x_{p}v, \xi(x) \right\}({\rm d} \xi(x))^2
\]
 in any local coordinate $\xi$. Here
 \[
 \{f, \xi \}=\left(\frac{f''(\xi)}{f'(\xi)}\right)' -\frac{1}{2}\left(\frac{f''(\xi)}{f'(\xi)}\right)^2
 \] -- Schwarzian derivative. The pullback of $q$ to $\hat{\mathcal{C}}$ has residueless 6-order poles at~$(x_j)$.

\item The meromorphic function $b\colon \hat{\mathcal{C}} \times \hat{\mathcal{C}} \xrightarrow{} \CC P^1 $
\begin{equation*}
b(x,y)=\frac{B(x,y)}{v(x)v(y)},
\end{equation*}
skew-symmetric in both arguments, with simple poles at $(x_j)$ on $\hat{\mathcal{C}}$ with respect to each argument. Outside the branch points on the diagonal, it has a double pole with the following asymptotics in the coordinate $z(x)=\int^x_{x_1}v$:
\begin{equation*}
b(x,y)=\frac{1}{(z(x)-z(y))^2}+\frac{1}{3}q(x)+\frac{1}{6}q_z(x)(z(y)-z(x))+\cdots, \qquad y \xrightarrow{} x.
\end{equation*}

\item The meromorphic function $h\colon {\mathcal{C}} \times {\mathcal{C}} \xrightarrow{} \CC P^1 $,
\begin{equation*}
h(x,y)=\frac{B^2(x,y)}{Q(x)Q(y)}=b^2(x,y).
\end{equation*}
Its pullback to $\hat{\mathcal{C}} \times \hat{\mathcal{C}}\xrightarrow{} \CC P^1$ is symmetric and has residueless double poles at $(x_j)$ in both arguments.
\end{itemize}

It is convenient to introduce the periods $\mathcal{P}_{s_i}=\oint_{s_i}v$ for $s_i$ being an element from the canonical basis of~$H_{-}$:
\begin{equation*}
\{s_i\}^{\dim (H_-)}_{i=1}= \big\{\big\{a^-_j, b^-_j \big\}^{g^-}_{j=1}, \big\{t^-_k \big\}^{n}_{k=1} \big\}.
\end{equation*}
The dual basis $\{s^*_i\}$ is defined by the condition
\begin{equation*}
s^*_i \circ s_j = \delta_{ij}
\end{equation*}
and is given by
\begin{equation*}
\{s^*_i\}^{\dim (H_-)}_{i=1}= \big\{\bigl\{-2b^-_j, 2a^-_j \bigr\}^{g^-}_{j=1}, \big\{2\kappa^-_k \big\}^{n}_{k=1} \big\},
\end{equation*}
here $\kappa^-_j$ is a $1/2$ of the contour connecting poles $z^{(1)}_j$ with $z^{(2)}_j$ and skew-symmetric under the involution, not intersecting other contours.
The following variational formulas were derived in~\cite{Bertola_2017} for a holomorphic differential $v$. In our framework, for $v$ being meromorphic the same formulas apply since the proof does not rely on the presence of poles of $v$. Note that the variations of the functions depending on the point $x \in \hat{\mathcal{C}}$ are computed assuming that the coordinate $z(x)$ is independent of the moduli.

\begin{Proposition}
For arbitrary basis $\{s_j\}^{6g-6+2n}_{j=1}$ of $H_-$ and its dual basis $\{ s^{*}_{j} \}^{6g-6+2n}_{j=1}$ the following formulas hold on $\mathcal{Q}_{g,n}[\mathbf{r}]$:
\begin{gather}\label{51}\frac{\delta \Omega_{ij}}{\delta \mathcal{P}_{s}}=\frac{1}{2}\oint_{s^*}{f}_i {f}_j v,\\
\label{52}\frac{\delta f_j(x)} {\delta \mathcal{P}_s}\Big|_{z(x)={\rm const}}=\frac{1}{4 \pi {\rm i}}\oint_{s^*}{f}_j(t){b}(x,t)v(t),\\
\label{53}\frac{\delta b(x,y)}{\delta \mathcal{P}_s}\Big|_{z(x),z(y)={\rm const}}=\frac{1}{4 \pi {\rm i}}\oint_{s^*}{b}(x,t){b}(t,y)v(t),\\
\label{54}\frac{\delta q(x)}{\delta \mathcal{P}_s}\Big|_{z(x)={\rm const}}=\frac{3}{4 \pi {\rm i}}\oint_{s^*}{h}(x,t)v(t).\end{gather}
\end{Proposition}

\subsection{Bergman tau-function}
The Bergman tau-function $\tau_B$ on the moduli space of the curves was originally defined as a~higher genus generalization of the Dedekind eta function on elliptic surface. It appears in various context -- from isomonodromy deformations to
spectral geometry, Frobenius manifolds and random matrices, see the review~\cite{korotkin2020bergman}. In the setting of moduli
spaces of quadratic differentials the Bergman tau-function was originally discussed in \cite{korotkin2013tau} in holomorphic case, later in~\cite{Bertola_2020} for meromorphic quadratic differentials with second-order poles.

In our framework, we consider a moduli space $\mathcal{Q}_{g,n}$ of quadratic meromorphic differentials with second-order poles. The explicit formula and main properties of $\tau_B$ were outlined in~\cite{bertola2021wkb}. In the present context, we only need its defining differential equations and transformation under rescaling of the differential $Q$ by a constant.
\begin{Theorem}[\cite{bertola2021wkb}]
The Bergman tau-function $\tau_B$ satisfies the following system of differential equations on $\mathcal{Q}_{g,n}$:
\begin{equation}\label{55} \frac{\delta \log \tau_B}{\delta A_j }=\frac{1}{12 \pi {\rm i}}\oint_{b^-_i}qv, \qquad \frac{\delta \log \tau_B}{\delta B_j }=-\frac{1}{12 \pi {\rm i}}\oint_{a^-_j}qv,\end{equation}
for j=$1,\dots,3g-3+n$ and
\begin{equation}\label{56}\frac{\delta \log \tau_B}{\delta (2 \pi {\rm i} r_k) }=-\frac{1}{12 \pi {\rm i}}\int_{\kappa^-_k}\left(qv+\frac{1}{4r^2_k}v \right),\end{equation}
for $k=1,\dots,n$.
\end{Theorem}
The function $\tau_B$ satisfies the following homogeneity property, which will be important in the present context:
\begin{equation*}
\tau_B(\mathcal{C}, \kappa Q) = \kappa^{ \frac{5(2g-2+n)}
{72}} \tau_B(\mathcal{C}, Q) .
\end{equation*}
Equivalently, defining the Euler vector field via
\begin{equation}\label{58}
E=\sum^{g^-}_{j=1}\left(A_j \frac{\delta}{\delta A_j}+B_j \frac{\delta}{\delta B_j}\right)+\sum^n_{j=1}r_j \frac{\delta}{\delta r_j}, \end{equation}
we have that
\begin{equation*}
E \log{\tau_B}=\frac{5(2g-2+n)}{72}.
\end{equation*}
We also notice that on the stratum $\mathcal{Q}_{g,n}[\mathbf{r}]$ the differential $\delta \log \tau_B|_{r}$ is given by
\begin{equation}\label{60}
\delta \log \tau_B |_{r}=-\frac{1}{12 \pi {\rm i}}\sum^{g_-}_{j=1} \bigg[\bigg(\oint_{a^-_j}qv \bigg)\delta B_j - \bigg(\oint_{b^-_j}qv \bigg)\delta A_j \bigg].
\end{equation}

\section{Monodromy map}\label{section3}

\subsection{Monodromy symplectomorphism and Goldman bracket}
Consider the monodromy map
\begin{equation}\label{62}
\mathcal{F}_{(S)}\colon \ \mathcal{Q}_{g,n} \xrightarrow[]{} {\rm CV}_{g,n}\end{equation}
for the equation~\eqref{1} defined by~\eqref{5}.
The Goldman bracket on ${\rm CV}_{g,n}[\mathbf{m}]$ is defined as follows~\cite{Goldman_1984}. For two arbitrary loops $\gamma$ and $\tilde{\gamma}$
\begin{equation}\label{63}
\{\operatorname{tr}M_\gamma, \operatorname{tr}M_{\tilde{\gamma}} \}_G=\frac{1}{2}\sum_{p \in \gamma \circ \tilde{\gamma}}\nu(p)(\operatorname{tr}M_{\gamma_p \tilde{\gamma}} - \operatorname{tr}M_{\gamma_p \tilde{\gamma}^{-1}} ),
\end{equation}
where the monodromy matrices $M_\gamma, M_{\tilde{\gamma}} \in {\rm PSL}(2, \CC) $; $\gamma_p \tilde{\gamma}$ and $\gamma_p \tilde{\gamma}^{-1}$
 are paths obtained by resolving the intersection point $p$ in two different ways (see \cite{Goldman_1984}); $\nu(p)=\pm 1 $ is the contribution of the point p to the
intersection index of $\gamma$ and $\tilde{\gamma}$.

The following theorem was stated in \cite{bertola2021wkb} and it is a natural extension of the results proven in~\cite{Bertola_2017} for holomorphic potentials and in \cite{korotkin2018periods} for potentials with simple poles.

\begin{Theorem}[\cite{bertola2021wkb}]\label{Th_4}
For the Bergman projective connection $S_{B}$ \eqref{12} chosen to be the base projective connection $S$ the monodromy map $\mathcal{F}_{(S_B)}$ \eqref{62} of equation \eqref{1} is Poisson. Namely, it preserves homological bracket \eqref{41} and minus the Goldman bracket \eqref{63} between traces of monodromy matrices.
\end{Theorem}

The homological shear coordinates, which are the Darboux coordinates for Goldman brack\-et~\eqref{63}, can be constructed in the following way (see appendix of \cite{bertola2021wkb}, also \cite{2021,chekhov} for the details):
Assume that quadratic differential $Q$ is generic, i.e., it does not have any saddle connections (as in the
definition of the ``Gaiotto--Moore--Neitzke differential''~\cite{giatto}).
Then each horizontal trajectory given by $\operatorname{Im} \int^x_{x_1}v=0$, where $x_1$ is an arbitrary ``first'' zero, starting at a zero $x_j$ of $Q$ ends at one of the poles $z_k$, defining critical graph $\Gamma_Q$. Additionally, three horizontal
trajectories meet at each zero, determining three vertices of the triangle at the poles, therefore,
defining the triangulation $\Sigma_Q$ of $\mathcal{C}$. The dual graph with vertices at $x_j$ is denoted by $\Sigma^*_Q$. Notice that the number of edges of $\Sigma_Q$ equals to $6g-6+3n=\dim ({\rm CV}_{g,n})$. The Thurston shear coordinate is a value $\zeta_e \in \CC$ attached to each edge $e$ of the graph $\Sigma_Q$.

Using the graph $\Sigma^*_Q$ one defines a two-sheeted branch covering $\hat{\mathcal{C}}_{\Sigma_B}$ by assuming that all edges of $\Sigma^*_Q$ are branch cuts. To each coordinate $\zeta_e$ we assign skew-symmetric cycle $l_e$ on $\hat{\mathcal{C}}_{\Sigma_B}$, which is a double cover of the corresponding dual edge $e^* \in \Sigma^*_Q$. Lemma A.1 of \cite{bertola2021wkb} states that the Goldman Poisson brackets \eqref{63} between the coordinates $\zeta_e$ can be expressed via the intersection indices of the cycles $l_e$ as
\begin{equation*}
\{\zeta_e, \zeta_{e'}\}_G=\frac{1}{4}l_e \circ l_{e'}.
\end{equation*}
An observation \cite[Proposition A.6]{bertola2021wkb} that the double cover $\hat{\mathcal{C}}_{\Sigma_B}$ is holomorphically equivalent to the canonical double cover $\hat{\mathcal{C}}$, defined analytically by $v^2=Q$, allows us to view these brackets in terms of corresponding cycles on $\hat{\mathcal{C}}$. Then one can consider linear combinations with half-integer coefficients of cycles $l_e$, generating the elements $\big\{a^-_j, b^-_j, t^-_k\big\}$ of the homology group
$H_-$. Taking the same linear combinations of the elements $2\zeta_e$ one defines the homological shear coordinates $\big\{\rho_{a^-_j}, \rho_{b^-_j}, \rho_{t^-_k}\big\}$ with the following Poisson brackets:
\begin{equation}\label{homsc}
 \big\{\rho_{a^-_j},\rho_{b^-_k} \big\}_{G}=\frac{\delta_{jk}}{2}, \qquad \big\{\rho_{a^-_j},\rho_{a^-_k} \big\}_{G}=\big\{\rho_{b^-_j},\rho_{b^-_k} \big\}_{G}=0, \end{equation}
while $\rho_{t^-_k}$ lie in the center of Poisson algebra.
The coordinates $\rho_{t^-_k}$
are related to the monodromy eigenvalue as follows:
\begin{equation*}\rho_{t^-_k}=\log m_k.\end{equation*}
Thus, on the symplectic leaf ${\rm CV}_{g,n}[\mathbf{m}]$ for fixed values of $m_k$ the Goldman symplectic form is written as
\begin{equation*}
\Omega_{G}=\sum^{g^-}_{j=1} 2 \delta\rho_{a^-_j} \wedge \delta\rho_{b^-_j}.
\end{equation*}

\begin{Corollary}
The homological symplectic form $\Omega_{{\rm hom}}$ on $\mathcal{Q}_{g,n}[\mathbf{r}]$ is minus pullback of the Goldman symplectic form $\Omega_{G}$ by the map $\mathcal{F}_{(S_B)}$
\begin{equation*}\mathcal{F}_{(S_B)}^*\Omega_{G}=-\Omega_{{\rm hom}}. \end{equation*}
\end{Corollary}

\subsection{Admissible meromorphic projective connections}
The map between the space of quadratic differentials $\mathcal{Q}_{g,n}[\mathbf{r}]$ and the character variety ${\rm CV}_{g,n}[\mathbf{m}]$ essentially depends on the choice of the base projective connection on a Riemann surface $\mathcal{C}.$ To parametrize space of such connections we consider the holomorphic affine bundle $\mathbb{S}_{g,n}$ of meromorphic projective connections with at most simple poles at the punctures over the moduli space of closed curves $\mathcal{M}_{g,n}$. Theorem~\ref{Th_4} states that for the choice $S=S_B$ the monodromy map is symplectic. The naturally arising question is when the monodromy map with a fixed projective connection~$S$ other than $S_B$ is also a symplectomorphism.

\begin{Definition}
A holomorphic section $S$ of the affine bundle $\mathbb{S}_{g,n}$ is called \textit{admissible} if the homological symplectic structure on $\mathcal{Q}_{g,n}[\mathbf{r}]$ implies Goldman bracket on the character variety~${\rm CV}_{g,n}[\mathbf{m}].$
\end{Definition}
For any two choices of $S_0, S_1 \in \mathbb{S}_{g,n} $ we write the same equation in two ways
\begin{equation*}
\partial^2 \phi+\left(\frac{1}{2}S_0-Q_0 \right)\phi=0 \qquad \text{and}\qquad
\partial^2 \phi+\left(\frac{1}{2}S_1-Q_1 \right)\phi=0,
\end{equation*}
where both $Q_0$ and $Q_1 $ belong to $\mathcal{Q}_{g,n}[\mathbf{r}]$ and are related by
\begin{equation}\label{66}Q_1-Q_0=\frac{1}{2}(S_0 -S_1 ).\end{equation}
We have the following diagram of maps
\[
 \begin{tikzcd}
\mathcal{Q}_{g,n}[\mathbf{r}] \arrow[r, " Q_0 \xrightarrow{\mathcal{H}} Q_1"] \arrow[dr, "\mathcal{F}_{(S_0)}"']
& \mathcal{Q}_{g,n}[\mathbf{r}] \ar[d, "\mathcal{F}_{(S_1)}"]\\
& {\rm CV}_{g,n}[\mathbf{m}].
\end{tikzcd}
\]
Assuming that $S_0$ is admissible, the condition for $S_1$ to be also admissible (or \textit{equivalent} to $S_0$) is that the map
\begin{equation}\label{67}\mathcal{H}\colon \ Q_0 \xrightarrow{} Q_0 +\frac{1}{2}(S_0 -S_1 )\end{equation} is a symplectomorphism implying the coincidence of homological 2-forms calculated via the periods of $v_0$ and $v_1$, where $v_0^2=Q_0$, $v_1^2=Q_1 $ define canonical coverings with different conformal structures. The following proposition gives a condition for the map \eqref{67} to be a symplectomorphism, generalizing the results proven in \cite{Bertola_2017,korotkin2018periods}, where $Q_0$, $Q_1$ are both assumed to be holomorphic or with simple poles, respectively.
\begin{Proposition} \label{Prop_3}\quad
\begin{enumerate}\itemsep=0pt
\item[$1.$] Two meromorphic differentials $Q_0, Q_1 \in \mathcal{Q}_{n,g}[\textbf{r}]$ induce the same homological $2$-form on $\mathcal{Q}_{n,g}[\textbf{r}]$ if and only if the $1$-form $\Theta_{(S_0-S_1)}$, corresponding to family of quadratic differentials $S_0-S_1$ and locally defined on $\mathcal{M}_{g,n}$,
is closed, $\delta\Theta_{(S_0-S_1)}=0.$

\item[$2.$] The generating function of the symplectomorphism between the periods $\big(A^{(0)}_k,B^{(0)}_k\big)$ of $v_0$ and $\big(A^{(1)}_k,B^{(1)}_k\big)$ of $v_1$ for the chosen potentials \begin{equation}\label{69}\theta_0=\sum^{g^-}_{k=1}\big(B^{(0)}_k \delta A^{(0)}_k - A^{(0)}_k \delta B^{(0)}_k\big), \qquad \theta_1=\sum^{g^-}_{k=1}\big(B^{(1)}_k \delta A^{(1)}_k - A^{(1)}_k \delta B^{(1)}_k\big),\end{equation}
defined by
\begin{equation*}
\delta\mathcal{G}_{{\rm hom}}=\mathcal{H}^*\theta_1 - \theta_0 ,
\end{equation*}
has the following form:
\begin{equation}\label{71}
\mathcal{G}_{{\rm hom}}=\sum^n_{i=1}\pi {\rm i} r_i\left({\rm reg}\int^{z^{(1)}_i}_{z^{(2)}_i}v_1 -{\rm reg}\int^{z^{(1)}_i}_{z^{(2)}_i}v_0 \right)+\frac{1}{2}\Hat{G}_{(S_0 -S_1)},
\end{equation}
where there exists a local holomorphic function $\Hat{G}_{(S_0 -S_1)}$ on $\mathcal{M}_{g,n}$, such that
\begin{equation*}
\delta\Hat{G}_{(S_0 -S_1)}= \Theta_{(S_0-S_1)}.
\end{equation*}
\end{enumerate}
\end{Proposition}

\begin{Remark}
For any $Q \in \mathcal{Q}_{g,n}[\mathbf{r}]$ the integral $\int^{z^{(1)}_i}_{z^{(2)}_i} v$ is singular at the endpoints. We define its regularization by removing the divergent part as follows: fix a coordinate $\xi_j$ near $z_j$ such that
\begin{equation*}
Q(x) \sim \left(\frac{r_j^2}{\xi_j^2} +O\big(\xi_j^{-1}\big)\right)({\rm d} \xi_j)^2.\end{equation*}
$\xi_j$ can also serve as a local coordinate on $\hat{\mathcal{C}}$ near the lifts $\big\{ {z^{(1)}_j}, {z^{(2)}_j}\big\}$ with
\begin{equation*}
v(x)\sim \pm \left(\frac{ r_j}{\xi_j} +O(1)\right){\rm d} \xi_j.\end{equation*}
Let $z^t_j$ be an arbitrary sequence of points on $\mathcal{C}$ converging to $z_j,$ such that in the local coordinate~$\xi_j$
\begin{equation*}
\operatorname{Re} \big(\xi_j\big(z^t_j\big)\big) \sim \frac{1}{t}, \quad t \xrightarrow{} \infty, \qquad \operatorname{Im} \big(\xi_j\big(z^t_j\big)\big)=0.\end{equation*}
Then the regularization is defined by
\begin{equation}\label{76}
{\rm reg}\int^{z^{(1)}_j}_{z^{(2)}_j}v := \lim_{t \xrightarrow{} \infty } \left( \int^{z^{t (1)}_j}_{z^{t (2)}_j}v - 2 r_j \log t \right).
\end{equation}
\end{Remark}
Before proceeding to the proof of Proposition \ref{Prop_3}, we will prove the following technical lemma, which will also be used in the computation of the WKB expansion of the Yang--Yang function. Introduce the pairing between any two meromorphic differentials $w_1$, $w_2$ on $\hat{\mathcal{C}}$:
\begin{equation}\label{77} \left\langle\oint w_1,\oint w_2 \right\rangle:=\sum^{g^-}_{j=1}\left[\oint_{b^-_j}w_1 \oint_{a^-_j}w_2 - \oint_{a^-_j}w_1 \oint_{b^-_j}w_2\right].\end{equation}

\begin{Lemma} \label{Lem_1}
Let $w_1$ and $w_2$ be two meromorphic differentials on $\hat{\mathcal{C}}$, skew-symmetric under involution. Assuming both $w_1$ and $w_2$ holomorphically depend on moduli $(A_i, B_i)^{g^-}_{i=1}$, the following identity of $1$-forms on $\mathcal{Q}_{g,n}[\textbf{r}]$ holds:
\begin{equation}\label{78}\left\langle \oint w_1, \delta \oint w_2\right\rangle =-\frac{1}{2}\int_{\partial \hat{\mathcal{C}}_0}\left(\delta w_2
\int^x_{p_0} w_1\right)+\left\langle \oint \frac{w_1 w_2}{v},\delta \oint v\right\rangle \end{equation} or
\begin{equation}\label{79}\left\langle \oint w_1, \delta \oint w_2\right\rangle =\frac{1}{2}\int_{\partial \hat{\mathcal{C}}_0}\left( w_1
\int^x_{p_0} \delta w_2\right)+\left\langle \oint\frac{w_1 w_2}{v},\delta \oint v\right\rangle. \end{equation}
Here $\hat{\mathcal{C}_0}$ is the fundamental polygon of the covering surface $\hat{\mathcal{C}}$, $p_0$ is a generic point.
\end{Lemma}
\begin{Remark}
Note that by $\delta \oint w$ appearing in the pairing we mean differential applied to the periods of $w$, while $\delta w$ is given by
\begin{equation}\label{80}\delta w:=\sum^{g^-}_{i=1} \left(\frac{\delta w}{\delta A_i}\Big|_{z(x)} \delta A_i +\frac{\delta w}{\delta B_i}\Big|_{z(x)} \delta B_i\right),\end{equation}
where the partial derivatives are defined under the assumption that the local coordinate~\eqref{36} is independent of the moduli. For example,
\begin{equation*}\frac{\delta w}{\delta A_i}\Big|_{z(x)}:=
v(x) \frac{\delta}{\delta A_i}\Big|_{z(x)={\rm const}}\left\{\frac{w(x)}{v(x)}\right\}, \end{equation*}
where $w(x)/v(x)$ is
a meromorphic function on $\hat{\mathcal{C}}$.
\end{Remark}
\begin{proof}
Expressing the differential in coordinates $(A_i, B_i)^{g^-}_{i=1}$, we write the pairing on the left-hand side of \eqref{78} as follows:
\begin{gather}
\sum^{g^-}_{j=1} \left(\oint_{b^-_j} w_1 \delta\oint_{a^-_j}w_2 - \oint_{a^-_j}w_1 \delta\oint_{b^-_j}w_2 \right)\nonumber\\
\label{82}\qquad{}=\sum^{g^-}_{i=1}\Bigg[ \sum^{g^-}_{j=1}\left(\oint_{b^-_j} w_1 \frac{\delta}{\delta A_i}\oint_{a^-_j}w_2-\oint_{a^-_j} w_1 \frac{\delta}{\delta A_i}\oint_{b^-_j}w_2\right)\delta A_i \\
\label{83}\qquad\quad{} +\sum^{g^-}_{j=1}\left(\oint_{b^-_j} w_1 \frac{\delta}{\delta B_i}\oint_{a^-_j}w_2-\oint_{a^-_j} w_1 \frac{\delta}{\delta B_i}\oint_{b^-_j}w_2\right)\Bigg]\delta B_i. \end{gather}

Take a reference point $p_0$ and consider a canonical dissection of the covering surface along the cycles in $H_1(\hat{\mathcal{C}})$ to obtain the fundamental polygon $\hat{\mathcal{C}_0}$. The coordinate $z(x)=\int^x_{p_0}v$ serves as a local coordinate on~$\hat{\mathcal{C}_0}$ outside branch points and poles and is kept fixed while differentiating with respect to the moduli. Consider the expression~\eqref{82} near~$\delta A_i$. According to the discussion in~\cite{Kokotov_2009}, when we differentiate the integral over $a^-_i$ with respect to the variable $A_i=\oint_{a^-_i}v$ an additional term appears:
\begin{equation*}
\frac{\delta}{\delta A_i}\oint_{a^-_i}w_2=\frac{w_2}{v}(R_i)+\oint_{a^-_i}\frac{\delta w_2}{\delta A_i},
\end{equation*}
where $R_i$ is an intersection point of the cycles $a^-_i$ and $2b^-_i$ (due to intersection index $a_i^- \circ b_j^- =\frac{\delta_{ij}}{2}$), whereas all other integrals commute with the differentiation with respect to coordinate $A_i.$ In our case $R_i=p_0.$ Therefore, we can rewrite the term near $\delta A_i$ as
\begin{equation}\label{85}
\sum^{g^-}_{j=1}\left(\oint_{b^-_j} w_1 \oint_{a^-_j}\frac{\delta w_2}{\delta A_i}-\oint_{a^-_j} w_1 \oint_{b^-_j}\frac{\delta w_2}{\delta A_i}\right)+\frac{w_2}{v}(p_0)\oint_{b^-_j} w_1.
\end{equation}
Let us write
$w_2(x)=f(z){\rm d}z$, for $x \in \hat{\mathcal{C}_0}$. Since $w_2$ is globally defined on $\hat{\mathcal{C}}$,
it is invariant under analytic continuation along the cycles in $H_{-}(\hat{\mathcal{C}})$.
Writing $w_2(x+a^-_i)=w_2(x)$, we have that
\begin{equation}\label{fzBfz}
f(z+A_i)=f(z).
\end{equation}
Differentiating this equality with respect to $z$, we get
\begin{equation*}
\frac{\partial f(z+A_i)}{\partial (z+A_i)}=\frac{\partial f(z)}{\partial z}.
\end{equation*}
Differentiating \eqref{fzBfz} again with respect to $A_i$, while $z$ is kept constant, we also mind that $f$ implicitly depends on $A_i$:
\begin{equation*}
\frac{\partial f(z+A_i)}{\partial (z+A_i)}+\frac{\delta f(z+A_i)}{\delta A_i}=\frac{\delta f(z)}{\delta A_i}.
\end{equation*}
Combining these formulas, we write
\begin{equation*}
\frac{\delta f(z+A_i)}{\delta A_i}{\rm d}z-\frac{\delta f(z)}{\delta A_i}{\rm d}z=-\frac{\partial f(z)}{\partial z}{\rm d}z,
\end{equation*}
or in invariant form (using that away from branch points and poles $v={\rm d}z)$
\begin{equation*}
\frac{\delta w_2}{\delta A_i}(x+a^-_i)-\frac{\delta w_2}{\delta A_i}(x)=-{\rm d} \left(\frac{w_2}{v}\right).
\end{equation*}
Hence, the differential $\frac{\delta}{\delta A_i}w_2$ could be seen as meromorphic on $\hat{\mathcal{C}}$ with a jump discontinuity $-{\rm d}\big(\frac{w_2}{v}\big)$ on the cycle $2b^-_i.$ Denote by $F:=\int^x_{p_0}w_1$. We apply a modification of the Riemann bilinear identity for differentials having discontinuities along the homology cycles. Splitting the integral over the boundary of $\hat{\mathcal{C}}_0$ into even and odd parts of $H_1\big(\hat{\mathcal{C}}, \Z\big)$ and recalling that the intersection index is $a_i^+ \circ b_j^+=a_i^- \circ b_j^- =\frac{\delta_{ij}}{2}$, one has
\begin{gather}\int_{\partial \hat{\mathcal{C}}_0}F\frac{\delta w_2}{\delta A_i} = \sum^{g^-}_{j=1}\Bigg[\left(\oint_{ 2b^-_j}F \frac{\delta w_2}{\delta A_i}+\oint_{( 2b^-_j)^{-1}}F\frac{\delta w_2}{\delta A_i} \right)+\left(\oint_{ a^-_j}F \frac{\delta w_2}{\delta A_i}+\oint_{( a^-_j)^{-1}}F \frac{\delta w_2}{\delta A_i} \right)
\nonumber\\
\hphantom{\int_{\partial \hat{\mathcal{C}}_0}F\frac{\delta w_2}{\delta A_i} =}{}
+\left(\oint_{ 2b^+_j}F \frac{\delta w_2}{\delta A_i}+\oint_{( 2b^+_j)^{-1}}F\frac{\delta w_2}{\delta A_i} \right)+\left(\oint_{ a^+_j}F \frac{\delta w_2}{\delta A_i} +\oint_{( a^+_j)^{-1}}F \frac{\delta w_2}{\delta A_i} \right)\Bigg].\label{86}
\end{gather}
Consider the following term of the above sum:
\begin{equation*}
\oint_{ 2b^-_i}F \frac{\delta w_2}{\delta A_i}+\oint_{( 2b^-_i)^{-1}}F \frac{\delta w_2}{\delta A_i}. \end{equation*}
It could be rewritten as
\begin{equation}\label{89}
\oint_{ 2b^-_i}\frac{\delta w_2}{\delta A_i}(P)\int^P_{p_0}w_1-\oint_{ 2b^-_i}\frac{\delta w_2}{\delta A_i}(P')\int^{P'}_{p_0}w_1,
\end{equation} where $P$, $P'$ are identified points on $ 2b^-_i$ and $(2 b^-_i)^{-1}$ cycles, respectively. $P'=P-a^-_i$.
That means that $P$, $P'$ lie on the different sides of a cycle $ 2b^-_i$, where $\frac{ \delta w_2}{\delta A_i}$ gains a jump. Then
\begin{equation*}
\frac{\delta w_2}{\delta A_i}(P')=\frac{\delta w_2}{\delta A_i}(P)-\text{``jump''}=\frac{\delta w_2}{\delta A_i}(P)+{\rm d}\left(\frac{w_2}{v}\right)(P).\end{equation*}
Plugging it into \eqref{89}, we rewrite this expression as
 \begin{equation*}
 \oint_{ 2b^-_i}\frac{\delta w_2}{\delta A_i}\oint_{a^-_i}w_1-\oint_{ 2b^-_i}{\rm d}\left(\frac{w_2}{v}\right)(P)\int^{P}_{p_0}w_1.
 \end{equation*}
Integrating the second term by parts, we get
 \begin{equation*}
 \oint_{ 2b^-_i}F \frac{\delta w_2}{\delta A_i}+\oint_{( 2b^-_i)^{-1}}F \frac{\delta w_2}{\delta A_i}=\oint_{ 2b^-_i}\frac{\delta w_2}{\delta A_i}\oint_{a^-_i}w_1-\frac{w_2}{v}(p_0)\oint_{ 2b^-_i}w_1+\oint_{ 2b^-_i}\frac{w_1 w_2}{v}.
 \end{equation*}
In all the the remaining terms of \eqref{86} the differential $\frac{\delta w_2}{\delta A_i}$ does not gain jump discontinuities and they could be commonly expressed as
\begin{gather}
\oint_{ 2b^-_j}F \frac{\delta w_2}{\delta A_i}+\oint_{( 2b^-_j)^{-1}}F \frac{\delta w_2}{\delta A_i}=\oint_{ 2b^-_j}\frac{\delta w_2}{\delta A_i}\oint_{a^-_j}w_1, \qquad j \neq i, \nonumber\\
\oint_{ a^-_j}F \frac{\delta w_2}{\delta A_i}+\oint_{( a^-_j)^{-1}}F \frac{\delta w_2}{\delta A_i}=-\oint_{ a^-_j}\frac{\delta w_2}{\delta A_i}\oint_{2b^-_j}w_1, \qquad \forall j, \nonumber\\
 \label{95}\oint_{ 2b^+_j}F \frac{\delta w_2}{\delta A_i}+\oint_{( 2b^+_j)^{-1}}F \frac{\delta w_2}{\delta A_i}=\oint_{ 2b^+_j}\frac{\delta w_2}{\delta A_i}\oint_{a^+_j}w_1, \qquad \forall j, \\
\label{96}\oint_{ a^+_j}F \frac{\delta w_2}{\delta A_i}+\oint_{( a^+_j)^{-1}}F \frac{\delta w_2}{\delta A_i}=-\oint_{ a^+_j}\frac{\delta w_2}{\delta A_i}\oint_{2b^+_j}w_1, \qquad \forall j. \end{gather}
The integrals in \eqref{95}, \eqref{96} over $a^+$, $b^+$ cycles vanish due to skew symmetry of $w_1$.
Thus, \eqref{86} could be rewritten as
\begin{gather}
\frac{1}{2}\int_{\partial \hat{\mathcal{C}}_0}F\frac{\delta w_2}{\delta A_i}=-\sum^{g^-}_{j=1}\left(\oint_{b^-_j} w_1
\oint_{a^-_j}\frac{\delta w_2}{\delta A_i}-\oint_{a^-_j} w_1
 \oint_{b^-_j}\frac{\delta w_2}{\delta A_i}\right)\nonumber\\
 \hphantom{\frac{1}{2}\int_{\partial \hat{\mathcal{C}}_0}F\frac{\delta w_2}{\delta A_i}=}{}
 -\frac{w_2}{v}(p_0)\oint_{ b^-_i}w_1+\oint_{ b^-_i}\frac{w_1 w_2}{v}.\label{97}
\end{gather}
Comparing the expressions \eqref{85} and \eqref{97}, we see that
\begin{gather*}
\left[ \sum^{g^-}_{j=1}\left(\oint_{b^-_j} w_1 \frac{\delta}{\delta A_i}\oint_{a^-_j}w_2-\oint_{a^-_j} w_1 \frac{\delta }{\delta A_i}\oint_{b^-_j}w_2\right)\right]\delta A_i\\
\qquad{} =\left[-\frac{1}{2}\int_{\partial \hat{\mathcal{C}}_0}F\frac{\delta w_2}{\delta A_i}+\oint_{ b^-_i}\frac{w_1 w_2}{v}\right] \delta A_i \end{gather*}
Similarly, one can show that
\begin{gather*}
\left[ \sum^{g^-}_{j=1}\left(\oint_{b^-_j} w_1 \frac{\delta}{\delta B_i}\oint_{a^-_j}w_2-\oint_{a^-_j} w_1 \frac{\delta}{\delta B_i}\oint_{b^-_j}w_2\right)\right]\delta B_i\\
\qquad{} =\left[-\frac{1}{2}\int_{\partial \hat{\mathcal{C}}_0}F\frac{\delta w_2}{\delta B_i}-\oint_{ a^-_i}\frac{w_1 w_2}{v}\right] \delta B_i.
 \end{gather*}
Plugging these expressions into \eqref{82}, \eqref{83}, one obtains the formula \eqref{78}. \eqref{79} follows from~\eqref{78} by applying the Stokes' theorem and the fact that in the interior of $\hat{\mathcal{C}_0}$ away from poles, when differentiating with respect to any local coordinate $\xi$, one has
\begin{equation*}
{\rm d}_\xi\left( \frac{\delta w_2}{\delta \mathcal{P} }\int^x_{p_0} w_1\right)=\frac{\delta w_2}{\delta \mathcal{P} } \wedge w_1 =- w_1 \wedge \frac{\delta w_2}{\delta \mathcal{P} }=-{\rm d}_\xi\left(w_1 \int^x_{p_0}\frac{\delta w_2}{\delta \mathcal{P} }\right)=0,
\end{equation*}
where $\mathcal{P} \in (A_i, B_i)^{g^-}_{i=1} $
\end{proof}

\begin{proof}[Proof of Proposition \ref{Prop_3}]
\textit{Case of two nearby differentials} $Q_0, Q_1 \in \mathcal{Q}_{g,n}[\mathbf{r}] $: let $U$ be a~sim\-ply-connected neighborhood of $Q_0$ and take $Q_1 \in U.$ For sufficiently small $\hbar$ this differential may be expressed as $Q_1=Q_0+\hbar \tilde{Q},$ where $\tilde{Q}$ is a quadratic differential with at most simple poles at~$(z_j)^{n}_{j=1}$. The canonical cover $\hat{\mathcal{C}_\hbar}$, defined by $\big(v^\hbar_1\big)^2=Q_1$, becomes $\hbar$-dependent. Consider $\mathcal{P}_{\hbar}$ to be one of the periods $\big(A^{(1)}_i,B^{(1)}_i\big)$ of $v^\hbar_1$. Then its $k$-th derivative with respect to $\hbar$ is given~by\footnote{Double factorial is defined by $n!!:=n(n-2)(n-4)\cdots$.}
\begin{equation}\label{101}
\frac{\partial^k }{\partial \hbar^k}\mathcal{P}_{\hbar}\Big|_{\hbar=0} =
(-1)^{k+1}\frac{(2k-3)!!}{2^k}\oint_{s}\frac{\tilde{Q}^k}{v_{0}^{2k-1}},\end{equation}
where $s$ is an element of $H_{-}(\hat{\mathcal{C}}_0, \Z)$.

To justify these formulas consider $ \mathcal{P}_{\hbar}=\oint_{s(\hbar)}v^\hbar_{1}$ as the integral on the base curve $\mathcal{C}$ via the projection $\pi\colon \hat{\mathcal{C}}_\hbar \xrightarrow{} {\mathcal{C}}$. If cycle $s(\hbar)$ belongs to the subset~\eqref{32}, using skew-symmetry of $v^\hbar_1$ it projects onto one of homology cycles of $\mathcal{C}$, which is independent of $\hbar$. If $s(\hbar)$ is an element of~\eqref{33}, it projects onto the contour encircling or passing through branch cuts arranged between pairs of zeroes of $Q^\hbar_{1}$. Despite the positions of the branch points vary along with $\hbar$, we may assume that the projections $\pi(s(\hbar))$ are kept fixed on $\mathcal{C}$. Then one has
\begin{equation*}
\frac{\partial}{\partial \hbar}\oint_{s(\hbar)}v^\hbar_{1}=\frac{\partial}{\partial \hbar}\oint_{\pi(s(\hbar))}\sqrt{Q_0+\hbar \tilde{Q}}= \oint_{\pi(s(\hbar))}\frac{\partial}{\partial \hbar}\sqrt{Q_0+\hbar\tilde{Q}},
\end{equation*}
and the differentiation is followed by pullback to $\hat{\mathcal{C}}$. Higher derivatives are obtained the same way.

Applying this argument, we can expand period coordinates by powers of $\hbar.$ Write
\begin{align*}
\mathcal{P}_{\hbar}& = \mathcal{P}_{\hbar}\Big|_{\hbar=0}+ \frac{\partial}{\partial \hbar}\mathcal{P}_{\hbar}\Big|_{\hbar=0}\hbar+\dots +\frac{\partial^k}{\partial \hbar^k}\mathcal{P}_{\hbar}\Big|_{\hbar=0}\frac{\hbar^k}{k!}+\cdots\\
& =\oint_{s} {v_0} + \frac{\hbar}{2}\oint_{s} \frac{\tilde{Q}}{ v_0}+\dots +\hbar^k\binom{\frac{1}{2}}{k}\oint_{s}\frac{\tilde{Q}^k}{v_0^{2k-1}}+\cdots.
\end{align*}
Plugging $\hbar$-expansions of the periods $\big(A^{(1)}_i,B^{(1)}_i\big)$ into the potential $\theta_1$ \eqref{69} and arranging terms by powers of $\hbar$, we write with the help of the pairing notation \eqref{77}:
\begin{gather}\label{105}
\theta_1=\left\langle \oint\! v_{1}, \delta \oint\! v_{1}\right\rangle =\left\langle \oint\! v_0, \delta \oint\! v_0\right\rangle + \hbar \left[\frac{1}{2}\left\langle \oint \!v_0, \delta \oint\! \frac{\tilde{Q}}{v_0}\right\rangle +\frac{1}{2}\left\langle \oint\! \frac{\tilde{Q}}{v_0}, \delta \oint\! v_0\right\rangle \right]\!\!\!\\
\hphantom{\theta_1=}{}
+\sum^{\infty}_{k=2}\hbar^k\Bigg[\binom{\frac{1}{2}}{k}\left(\left\langle \oint v_0,\delta \oint \frac{\tilde{Q}^k}{v_0^{2k-1}}\right\rangle +\left\langle \oint \frac{\tilde{Q}^k}{v_0^{2k-1}},\delta \oint v_0\right\rangle \right)\nonumber
\\
\hphantom{\theta_1=}{}
+\sum^{k-1}_{l=1} \binom{\frac{1}{2}}{l} \binom{\frac{1}{2}}{k-l}
\left\langle \oint \frac{\tilde{Q}^l}{v_0^{2l-1}}, \delta \oint \frac{\tilde{Q}^{k-l}}{v_0^{2(k-l)-1}} \right\rangle \Bigg].\label{106}
\end{gather}

We will treat separately the expressions near $\hbar^1$ and $\hbar^k$, $k \geq 2$.

\textit{Coefficient near $\hbar^1$}: noticing that
\begin{equation*}
\delta\left\langle \oint v_0, \oint \frac{\tilde{Q}}{v_0}\right\rangle =\left\langle \delta \oint v_0, \oint \frac{\tilde{Q}}{v_0}\right\rangle +\left\langle \oint v_0, \delta \oint \frac{\tilde{Q}}{v_0}\right\rangle,
\end{equation*}
the expression near $\hbar^1$ could be rewritten as
\begin{equation*}
\frac{1}{2}\delta\left\langle \oint v_0, \oint \frac{\tilde{Q}}{v_0}\right\rangle +\left\langle \oint \frac{\tilde{Q}}{v_0}, \delta \oint v_0\right\rangle .
\end{equation*}
 Applying the Riemann bilinear identity, the pairing $\big\langle \oint v_0, \oint \frac{\tilde{Q}}{v_0}\big\rangle $ could be written as the sum over residues inside the fundamental domain. Differential $\tilde{Q}$, being lifted to $\hat{\mathcal{C}}$, gains double zeroes at branch points $x_i$, which are the zeroes of $Q_0$, and simple poles at preimages of~$z_j$. That makes $\frac{\tilde{Q}}{v_0}$ holomorphic, while $v_0$ has poles at
$\big\{z^{(1)}_j, z^{(2)}_j \big\}^n_{j=1}$. Therefore, using that near $z^{(1)}_j$ \big($z^{(2)}_j$\big) in the local coordinate~\eqref{39} $v_0=\pm \frac{r_j}{\zeta_i}{\rm d}\zeta_i$, we write
\begin{equation}\label{109}\delta
\left[\sum^n_{j=1}\pi {\rm i} \underset{\{z^{(1)}_j, z^{(2)}_j \}}{\rm res}\left(v_0
\int^x_{p_0}\frac{\tilde{Q}}{v_0} \right)\right]=\delta\left[\sum^n_{j=1}\pi{\rm i} r_j
\int^{z^{(1)}_j}_{z^{(2)}_j}\frac{\tilde{Q}}{v_0} \right]. \end{equation}
We will also express the pairing $\big\langle \oint \frac{\tilde{Q}}{v_0}, \delta \oint v_0\big\rangle $ in a different form, introducing the system of local coordinates on $\mathcal{M}_{g,n}$. For simplicity here we restrict us to the case $g \geq 2$ (low genus cases $g=0,1$ could be covered by analogy following~\cite{korotkin2018periods}). At generic point of the moduli space~$\mathcal{M}_{g,n}$ the quadratic differential~$\tilde{Q}$ could be represented as a linear combination of $3g-3$ products of normalized holomorphic differentials $u_j u_k$, where $(jk) \in D$ for some subset~$D$ of entries of matrix~$\Omega$, and additional~$n$ quadratic differentials encoding the meromorphic part could be represented by the following generically meromorphic differentials~$Q_k$ whose only pole of order one located at~$z_k$:
\begin{equation}\label{110}
Q^{z_k}(t)=\frac{1}{4 \pi {\rm i}}\frac{u_i(t)u_j(z_k)-u_i(z_k)u_j(t)}{u_j^2(z_k)}B(t, z_k),
\end{equation}
here $u_i$ and $u_j$ are two arbitrary normalized holomorphic differentials such that $u_j(z_k) \neq 0$.

Using the variational formulas \eqref{52} after lifting the function $\frac{u_i}{u_j}(x)$ to $\hat{\mathcal{C}}$, we get
\begin{equation}\label{111}
\frac{\delta }{\delta \mathcal{P}_s}\left[\frac{u_i}{u_j}(z_k)\right]=\frac{\delta }{\delta \mathcal{P}_s}\Bigg[\frac{u_i}{u_j}\big(z^{(1)}_k\big)\Bigg]=\oint_{s^*}\frac{Q^{z^{(1)}_k}}{v_0},
\end{equation}
where $Q^{z^{(1)}_k}$ is a lift of $Q^{z_k}$ to $\hat{\mathcal{C}}$. The entries $\Omega_{jk}, \ (jk) \in D$ of the period matrix can serve as the moduli of the base curve $\mathcal{C},$ while $\frac{u_i}{u_j}(z_k):=q_k$ code the positions of poles, providing in total $3g-3+n$ local coordinates on $\mathcal{M}_{g,n}.$

At generic point of $\mathcal{M}_{g,n}$ quadratic differential $\tilde{Q}$ can be expressed as
\begin{equation*}
\tilde{Q}=\sum_{(jk) \in D}p_{jk} u_j u_k+ \sum^n_{l=1} p_l Q^{z_l}, \qquad p_{jk}, p_l \in \CC.\end{equation*}
Then applying variational formulas \eqref{51} and \eqref{111}, one has
\begin{gather}
\left\langle \oint \frac{\tilde{Q}}{v_0}, \delta \oint v_0\right\rangle =\sum^{g_-}_{j=1} \left[ \left(\oint_{b^-_j}\frac{\tilde{Q}}{v_0} \right)\delta A^{(0)}_j - \left(\oint_{a^-_j}\frac{\tilde{Q}}{v_0} \right)\delta B^{(0)}_j \right] \nonumber\\
\qquad{}
 =\sum_{(jk) \in D}p_{jk}\sum^{g_-}_{j=1} \left[ \left(\oint_{b^-_j}\frac{u_j u_k}{v_0} \right)
 \delta A^{(0)}_j - \left(\oint_{a^-_j}\frac{u_j u_k}{v_0} \right)\delta B^{(0)}_j \right ]\nonumber\\
 \qquad\quad{} + \sum^n_{l=1} p_l \sum^{g_-}_{j=1} \left[ \left(\oint_{b^-_j}\frac{Q^{z^{(1)}_l}}{v_0} \right)\delta A^{(0)}_j
 - \left(\oint_{a^-_j}\frac{Q^{z^{(1)}_l}}{v_0} \right)\delta B^{(0)}_j \right ]\nonumber\\
\qquad{}
=\sum_{(jk) \in D}\!p_{jk}\sum^{g_-}_{j=1} \!\left[ \frac{\partial \Omega_{ij}}{\partial A^{(0)}_j} \delta A^{(0)}_j +\frac{\partial \Omega_{ij}}{\partial B^{(0)}_j}\delta B^{(0)}_j \right]+\sum^n_{l=1} p_l \sum^{g_-}_{j=1} \!\left[ \frac{\partial q_l}{\partial A^{(0)}_j}\delta A^{(0)}_j + \frac{\partial q_l}{\partial B^{(0)}_j}\delta B^{(0)}_j \right ]
\nonumber\\
\qquad{}
=\sum_{(jk) \in D}p_{jk} \delta \Omega_{jk}+ \sum^n_{l=1} p_l \delta q_l. \label{114}
\end{gather}
Therefore, the term near $\hbar^1$ becomes
\begin{equation}\label{115}\delta\left[\sum^n_{j=1}\frac{\pi {\rm i} r_j}{2}
\int^{z^{(1)}_j}_{z^{(2)}_j}\frac{\tilde{Q}}{v_0} \right]+\sum_{(jk) \in D}p_{jk} \delta \Omega_{jk}+ \sum^n_{l=1} p_l \delta q_l.\end{equation}

\textit{Coefficients near $\hbar^k$, $k \geq 2$}: by Lemma \ref{Lem_1}, we can rewrite the pairings appearing in \eqref{106} as
\begin{gather*}
\left\langle \oint \frac{\tilde{Q}^l}{v_0^{2l-1}}, \delta \oint \frac{\tilde{Q}^{k-l}}{v_0^{2(k-l)-1}}\right\rangle =\frac{1}{2}\int_{\partial \hat{\mathcal{C}}_0}\left( \frac{\tilde{Q}^l}{v_0^{2l-1}}
\int^x_{p_0}\delta \frac{\tilde{Q}^{k-l}}{v_0^{2(k-l)-1}}\right)+\left\langle \oint\frac{\tilde{Q}^k}{v_0^{2k-1}},\delta \oint v_0\right\rangle. \end{gather*}
Also
\begin{equation*}
\left\langle \oint v_0, \delta \oint \frac{\tilde{Q}^k}{v_0^{2k-1}}\right\rangle =\frac{1}{2}\int_{\partial \hat{\mathcal{C}}_0}\left(v_0
\int^x_{p_0}\delta\frac{\tilde{Q}^k}{v_0^{2k-1}}\right)+\left\langle \oint \frac{\tilde{Q}^k}{v_0^{2k-1}},\delta \oint v_0\right\rangle.
\end{equation*}
Thus, the expression \eqref{106} near $\hbar^k$, $k \geq 2$ becomes
\begin{gather}\label{118}
\binom{\frac{1}{2}}{k}\frac{1}{2}\int_{\partial \hat{\mathcal{C}}_0}\left(v_0
\int^x_{p_0}\delta\frac{\tilde{Q}^k}{v_0^{2k-1}}\right)
-\sum^{k-1}_{l=1}\binom{\frac{1}{2}}{l} \binom{\frac{1}{2}}{k-l}\frac{1}{2}\int_{\partial
\hat{\mathcal{C}}_0}\Bigg(\delta \frac{\tilde{Q}^{k-l}}{v_0^{2(k-l)-1}}
\int^x_{p_0}\frac{\tilde{Q}^l}{v_0^{2l-1}}\Bigg)\\
 \label{119}\qquad{} +\left[ \sum^{k}_{l=0}\binom{\frac{1}{2}}{l}\binom{\frac{1}{2}}{k-l}\right]
 \left\langle \oint \frac{\tilde{Q}^k}{v_0^{2k-1}},\delta \oint v_0\right\rangle .
 \end{gather}

Using the identity
\begin{equation*}
1+t=\big(\sqrt{1+t}\big)^2=\left(\sum^{\infty}_{k=0}\binom{\frac{1}{2}}{k}t^k \right)^2=\sum^{\infty}_{k=0}t^k \left[\sum^k_{l=0}\binom{\frac{1}{2}}{l} \binom{\frac{1}{2}}{k-l} \right], \qquad |t| \leq 1, \end{equation*}
and comparing the expressions near same powers of $t$ we conclude that the piece \eqref{119} is zero. Further, we can represent the expression in~\eqref{118} as the sum over residues at the branch points~$x_i$. Notice that the first term also has additional residues near $\big\{z^{(1)}_j, z^{(2)}_j \big\}^n_{j=1}$ due to simple poles of~$v_0$:
\begin{gather}\label{121}
\sum^n_{j=1}\binom{\frac{1}{2}}{k}\pi {\rm i}\underset{\{z^{(1)}_j, z^{(2)}_j \}}{\rm res}\left(v_0
\int^x_{p_0}\delta\frac{\tilde{Q}^k}{v_0^{2k-1}} \right)\\
 \label{122}
 {}+\sum^{4g-4+2n}_{i=1}\pi{\rm i} \, \underset{x_i}{\rm res}\left[\binom{\frac{1}{2}}{k}v_0
\int^x_{p_0}\delta\frac{\tilde{Q}^k}{v_0^{2k-1}}+ \sum^{k-1}_{l=1}\binom{\frac{1}{2}}{l} \binom{\frac{1}{2}}{k-l}\left(\frac{\tilde{Q}^l}{v_0^{2l-1}}
\int^x_{p_0}\delta \frac{\tilde{Q}^{k-l}}{v_0^{2(k-l)-1}}\right)\right].
\end{gather}
Similarly to \eqref{109}, the sum \eqref{121} could be rewritten as
\begin{equation}\label{123}\delta\left[\sum^n_{j=1}\binom{\frac{1}{2}}{k}
\pi {\rm i} r_j \int^{z^{(1)}_j}_{z^{(2)}_j}\frac{\tilde{Q}^k}{v_0^{2k-1}} \right],\end{equation}
here we used that the derivatives with respect to the coordinates $\big(A^{(0)}_i, B^{(0)}_i\big)^{g^-}_{i=1}$ commute with the integral.

Consider the expression defined on the double cover $\pi\colon \hat{\mathcal{C}}_0 \xrightarrow[]{} \mathcal{C}$, given by $v^2_0=Q_0$ in $T^*\mathcal{C}$:
\begin{equation*}
\underset{x_i}{\rm res}\left(\sqrt{Q_0+\hbar \tilde{Q}}\int^x_{p_0} \delta \sqrt{Q_0+\hbar \tilde{Q}}\right),
\end{equation*} where $x_i$ is a zero of $Q_0$ on~$\mathcal{C}$.

Formally, the Abelian differential $\hat{v}_\hbar=\sqrt{Q_0+\hbar \tilde{Q}}$ is globally defined on the $h$-dependent double cover $\hat{\pi}\colon \hathat{\mathcal{C}}_\hbar \xrightarrow[]{} \hat{\mathcal{C}}_0$, given by $(\hat{v}_\hbar)^2=Q_0+\hbar \tilde{Q}$ in $T^*\hat{\mathcal{C}}_0$ (note that $\hat{\mathcal{C}}_0$ itself is a double cover of $\mathcal{C}$).
Lifted from $\mathcal{C}$ to $\hat{\mathcal{C}}_0$, $Q_0$ has a 4th-order zeros at $x_i$ and double poles at $\big(z^{(1)}_j, z^{(2)}_j\big)^n_{j=1}$, while $\tilde{Q}$ gains a 2nd-order zero at $x_i$ and simple poles at $(z^{(1)}_j, z^{(2)}_j)$.
Thus, the map $\hat{\pi}$ is brached at $8g-8+4n$ simple zeroes $\tilde{x}^\hbar_j$ of ${Q_0+\hbar \tilde{Q}}$.
The double cover $\hathat{\mathcal{C}}_\hbar$ is smooth everywhere except for preimages of double zeroes $(x_i)^{4g-4+2n}_{i=1}$ of $Q_0+\hbar \tilde{Q}$, where $\hathat{\mathcal{C}}_\hbar$ gains nodes. The genus~$\hathat{g}$ of~$\hathat{\mathcal{C}}_\hbar$ equals $12g-11+4n$. Letting $\hbar \xrightarrow[]{} 0$, the nodes smoothen out and the covering surface~$\hathat{\mathcal{C}}_\hbar$ degenerates to the pair of smooth surfaces $\hat{\mathcal{C}}^{(1,2)}_0$. On the base curve $\hat{\mathcal{C}}_0$ that corresponds to the merging of triplets of points: two simple zeroes $\tilde{x}^\hbar_{i_1}$, $x^\hbar_{i_2}$ of ${Q_0+\hbar \tilde{Q}}$ converge to a double zero at~$x_i$, increasing its multiplicity to~4.

In the local coordinate $z(x)$ on $\hat{\mathcal{C}}_0$: $Q_0=v_0^2={\rm d}z^2$, $\tilde{Q}=\tilde{Q}(z){\rm d}z^2$. Differentiating with respect to the coordinates $\big(A^{(0)}_i, B^{(0)}_i\big)^{g^-}_{i=1}$ according to the rule \eqref{80}, when the coordinate $z(x)$ is kept fixed, one has that the residue could be written as
\begin{equation*}
\underset{x_i}{\rm res}\left(\sqrt{Q_0+\hbar \tilde{Q}}\int^x_{p_0} \frac{\hbar \delta \tilde{Q}}{2\sqrt{Q_0+\hbar \tilde{Q}}}\right)=0.
\end{equation*}
The residue vanishes since the expression inside is holomorphic at $x_i$. Then we can expand the left-hand side by powers of $\hbar$ and observe that the coefficients near the powers of $\hbar$ in the series are exactly the terms appearing in the sum \eqref{122}. It follows that these coefficients must vanish too. Thus, the coefficient near $\hbar^n, \ n \geq 2$ reduces to the expression \eqref{123}.

\textit{Full expansion}: combining \eqref{115} and \eqref{123}, we have that
\begin{gather}\label{126}\theta_1- \theta_0=\delta\left[\sum^n_{j=1}
\pi {\rm i} r_j \left[\sum^{\infty}_{k=1}\hbar^k\binom{\frac{1}{2}}{k}\int^{z^{(1)}_j}_{z^{(2)}_j}\frac{\tilde{Q}^k}{v_0^{2k-1}} \right]\right] + \hbar \left[\sum_{(jk) \in D}p_{jk} \delta \Omega_{jk}+ \sum^n_{l=1} p_l \delta q_l \right]. \end{gather}

We further notice, similarly to the argument in \eqref{101}, that the infinite series is formally the Taylor expansion by powers of $\hbar$ of the expression
\begin{equation*}
\int^{z^{(1)}_j(\hbar)}_{z^{(2)}_j(\hbar)}v^\hbar_1 - \int^{z^{(1)}_j}_{z^{(2)}_j}v_0 .
\end{equation*}
The issue is that the integrands are singular at the end points of the integration path, and one requires a regularization of the integral to have a proper identity.
Considering the regularization proposed in~\eqref{76}, one can see that
\begin{equation*}
{\rm reg}\int^{z^{(1)}_j(\hbar)}_{z^{(2)}_j(\hbar)}v^\hbar_1 ={\rm reg}\int^{z^{(1)}_j}_{z^{(2)}_j}v_0 + \sum^{\infty}_{k=1}\hbar^k \binom{\frac{1}{2}}{k}\int^{z^{(1)}_j}_{z^{(2)}_j}\frac{\tilde{Q}^k}{v_0^{2k-1}} ,
\end{equation*}
where the integrals near $h^n$ are already regular.
Using that, we rewrite the difference of potentials~as
\begin{gather}
\theta_1- \theta_0=\delta\left[\sum^n_{j=1}
\pi {\rm i} r_j \left({\rm reg}\int^{z^{(1)}_j(\hbar)}_{z^{(2)}_j(\hbar)}v^\hbar_1 - {\rm reg}\int^{z^{(1)}_j}_{z^{(2)}_j}v_0 \right)\right] \nonumber\\
\hphantom{\theta_1- \theta_0=}{}
+ \hbar \left[\sum_{(jk) \in D}p_{jk} \delta \Omega_{jk}+ \sum^n_{l=1} p_l \delta q_l \right].\label{129}
 \end{gather}

\textit{Case of two arbitrary differentials} $Q_0, Q_1 \in \mathcal{Q}_{g,n}[\mathbf{r}] $.
Theorem~1.3 of~\cite{chen2021classification} asserts that generically (outside of hyperelliptic locus for $g \geq 3$) space $\mathcal{Q}_{g,n}$, and, thus $\mathcal{Q}_{g,n}[\mathbf{r}],$ is connected. Let $\gamma_t=Q_t\colon [0,1] \xrightarrow{} \mathcal{Q}_{g,n}[\mathbf{r}] $ be a path such that $\gamma(0)=Q_0$, $\gamma(1)=Q_1.$ For each~$t$, the function $\oint_{s_i} v_t$ is holomorphic on $\mathcal{Q}_{g,n}[\mathbf{r}]$ and could be expanded by the Taylor series in some simply-connected open neighborhood $U_t$ of $Q_t$. Then $\bigcup_t U_t$ provides and open cover for $\gamma_t$. Due to compactness of $\gamma_t$, we can choose some finite subcover $\bigcup_{t_i} U_{t_i}$, $i \in [[0,N]].$ We can assume $Q_0=Q_{\hat{t}_0} \in U_{t_0}$, $Q_1=Q_{\hat{t}_{N+1}} \in U_{t_N}$ and take $Q_{\hat{t}_i} \in \gamma \cap U_{t_{i}} \cap U_{t_{i-1}}$, $i \in [[1,N]]$.
Due to~\eqref{129},
\begin{gather*}
\theta_{\hat{t}_{i+1}}-\theta_{\hat{t}_i}= \delta \left[\sum^n_{j=1} \pi {\rm i} r_j \left({\rm reg}\int^{z^{(1)}_j}_{z^{(2)}_j}{v_{\hat{t}_{i+1}}}-{\rm reg}\int^{z^{(1)}_j}_{z^{(2)}_j}{v_{\hat{t}_i}} \right) \right]\\
\hphantom{\theta_{\hat{t}_{i+1}}-\theta_{\hat{t}_i}=}{}
+\left[\sum_{(jk) \in D}p^{(\hat{t}_{i+1}, \hat{t}_i)}_{jk} \delta \Omega_{jk}+\sum^n_{j=1}p^{(\hat{t}_{i+1}, \hat{t}_i)}_l \delta q_l \right],
\end{gather*}
where $\big(p^{(\hat{t}_{i+1}, \hat{t}_i)}_{jk}, p^{(\hat{t}_{i+1}, \hat{t}_i)}_l\big)$ are coefficients of the linear representation of $\big(Q_{\hat{t}_{i+1}}-Q_{\hat{t}_{i}}\big)$ in the basis $\big(u_j u_k, Q^{z_l}\big)$.

Then applying the telescoping series, one has
\begin{gather}\theta_{1}-\theta_{0}=\sum^{N}_{i=0}(\theta_{\hat{t}_{i+1}}-\theta_{\hat{t}_i})= \delta \left[\sum^n_{j=1} \pi{\rm i} r_j \left({\rm reg}\int^{z^{(1)}_j}_{z^{(2)}_j}v_1-{\rm reg}\int^{z^{(1)}_j}_{z^{(2)}_j} v_0 \right) \right]\nonumber\\
\hphantom{\theta_{1}-\theta_{0}=}{}
+\left[\sum_{(jk) \in D}p^{(1,0)}_{jk} \delta \Omega_{jk}+\sum^n_{l=1}p^{(1,0)}_l \delta q_l\right],\label{131}
\end{gather}
where the latter expression, using \eqref{66}, is exactly the 1-form $\frac{1}{2}\Theta_{(S_0-S_1)}.$ By the assumption, coefficients $p^{(1,0)}_{jk}$ and $p^{(1,0)}_l$ are holomorphic functions of $(\Omega_{jk},q_l)$.
Applying differential to both sides of \eqref{131} one obtains the first statement of the proposition. If $\Theta_{(S_0 -S_1)}$ is assumed closed, then by the Poincar\'e lemma it could be locally integrated, leading to the second statement.
\end{proof}

Combining the result of the Theorem \ref{Th_4} and Proposition \ref{Prop_3}, we can formulate a condition for $S$ to become admissible.
\begin{Theorem} \label{Th_5}
The monodromy map
\begin{equation*}
\mathcal{F}_{(S)}\colon \ \mathcal{Q}_{g,n}[\mathbf{r}] \xrightarrow[]{} {\rm CV}_{g,n}[\mathbf{m}]\end{equation*}
is a symplectomorphism with $\mathcal{F}_{(S)}^*\Omega_G=-\Omega_{\hom}$
 if and only if the $1$-form $\Theta_{(S-S_B)}$, corresponding to family of quadratic differentials $S-S_B$ and locally defined on $\mathcal{M}_{g,n}$,
is closed, $\delta\Theta_{(S-S_B)}=0$. Equivalently, if and only if there exists a~local holomorphic function $\Hat{G}_{(S-S_B)}$ on $\mathcal{M}_{g,n}$, such that%
\begin{equation}\label{134}
\delta\Hat{G}_{(S-S_B)}=\Theta_{(S-S_B)}.
\end{equation}
\end{Theorem}

The computation similar to \eqref{114} (performed backwards) allows us to characterize the admissible projective connection in terms of the 1-form defined on $\mathcal{Q}_{g,n}[\mathbf{r}]$ in period coordinates.

\begin{Corollary}
The projective connection $S \in \mathbb{S}_{g,n}$ is admissible if and only if the following locally defined $1$-form on $\mathcal{Q}_{g,n}[\mathbf{r}]$
\begin{equation*}
\Theta_{(S-S_B)}=\sum^{g^-}_{j=1} \left [ \left(\oint_{b^-_j}\frac{S-S_B}{v} \right)\delta A_j - \left(\oint_{a^-_j}\frac{S-S_B}{v} \right)\delta B_j \right]
\end{equation*}
is closed, $\delta\Theta_{(S-S_B)}=0.$
\end{Corollary}

The following corollary gives an alternative characterization of admissible projective connections which
does not refer to the Bergman projective connection:

\begin{Corollary}The projective connection $S \in \mathbb{S}_{g,n}$ is admissible if and only if the following locally defined $1$-form on $\mathcal{Q}_{g,n}[\mathbf{r}]$
\begin{equation*}
\Theta_{(S-S_v)}=\sum^{g^-}_{j=1} \left[ \left(\oint_{b^-_j}\frac{S-S_v}{v} \right)\delta A_j - \left(\oint_{a^-_j}\frac{S-S_v}{v} \right)\delta B_j \right]
\end{equation*}
is closed, $\delta\Theta_{(S-S_v)}=0.$
\end{Corollary}
\begin{proof}
Notice, that from \eqref{60} it follows that 1-forms
$\Theta_{(S-S_v)}$ and $\Theta_{(S-S_B)}$ differ by the closed form $(24 \pi {\rm i})\delta \log \tau_B |_r$
so, their conditions of closedness are equivalent.
\end{proof}

In \cite{Bertola_2017}, authors discussed alternative ways of fixing the reference projective connection. It was showed that if $S$ is chosen to be either Schottky, Wirtinger or Bers projective connection, it is equivalent to the Bergman projective connection $S_B$ in the sense \eqref{134}. While explicit formulas~$\Hat{G}_{(S-S_B)}$ for Schottky and Wirtinger connections were derived in \cite{Bertola_2017}, for Bers connection it was only conjectured, and recently proven in \cite{https://doi.org/10.48550/arxiv.2109.02033}. Moreover,
the definition of Bergman projective connection itself depends on the choice of Torelli marking on~$\mathcal{C}.$ Let two Torelli markings $\alpha^\sigma$ and $\alpha$ be related by ${\rm Sp}(2g,\Z)$ matrix
\begin{equation}\label{matrix_hom}
\sigma=\begin{pmatrix}
C & A \\
D & B
\end{pmatrix}\colon \
\begin{pmatrix}
b \\
a
\end{pmatrix}^{\sigma}= \sigma
\begin{pmatrix}
b \\
a
\end{pmatrix}.
\end{equation}
Then two corresponding Bergman projective connections $S^{\sigma}_B$ and $S_B$ are related by
\begin{equation}\label{136}
S^{\sigma}_B=S_B-12 \pi {\rm i} \sum_{1 \leq j \leq k \leq g}u_j u_k \frac{\delta}{\delta \Omega_{jk}}\log \det (C \Omega +D).
\end{equation}
and also equivalent due \eqref{134} with the generating function $\Hat{G}_{(S^{\sigma}_B-S_B)}$ given by
\begin{equation}\label{gen}
\Hat{G}_{(S^{\sigma}_B-S_B)}=-12 \pi {\rm i} \log \det (C \Omega +D).
\end{equation}
That allows us to formulate the following corollary of Theorem~\ref{Th_5}.

\begin{Corollary}
If $S \in \mathbb{S}_{g,n}$ is chosen to be either Bergman $($corresponding to any Torelli marking$)$, Schottky, Wirtinger or Bers $($defined with respect their own data$)$ then the monodromy map
\begin{equation*}
\mathcal{F}_{(S)}\colon \ \mathcal{Q}_{g,n}[\mathbf{r}] \xrightarrow[]{} {\rm CV}_{g,n}[\mathbf{m}]
\end{equation*} is a symplectomorphism with $\mathcal{F}_{(S)}^*\Omega_G=-\Omega_{\hom}$.
\end{Corollary}

\subsection{Definition of the monodromy generating function}
Fixing the Bergman projective connection as the base connection $S=S_B$, we may choose a~symplectic potential on the moduli space $\mathcal{Q}_{g,n}[\mathbf{r}]$ in period coordinates
\begin{equation}\label{pot_hom}\theta_{{\rm hom}}=\sum^{g^-}_{j=1} \big(B_j \delta A_j - A_j \delta B_j \big)\end{equation}
with another symplectic potential on the character variety ${\rm CV}_{g,n}[\mathbf{m}]$ in homological shear coordinates
\begin{equation}\label{pot_G}
\theta_{G}=\sum^{g^-}_{j=1}\big( \rho_{b^-_j} \delta \rho_{a^-_j} - \rho_{a^-_j} \delta \rho_{b^-_j} \big)
\end{equation}
and consider the generating function of symplectomorphism $\mathcal{F}_{(S_B)}$ (the Yang--Yang function introduced in~\cite{nekrasov} ) given by
 \begin{equation} \label{yang_1}
 \delta \mathcal{G}_{B}=\mathcal{F}_{(S_B)}^*\theta_{G}-\theta_{{\rm hom}}.
 \end{equation}
Assuming that the triangulation of the surface $\mathcal{C}$ used to define homological shear coordinates is specified by the horizontal trajectories of the GMN-differential $Q$, the remaining parameters that define the function $\mathcal{G}_{B}$ include the choice of the Torelli marking on $\mathcal{C}$ and the choice of generators $\big(a^-_j,b^-_j\big)$ in $H_{-}.$ It is easy to see, that the symplectic potentials are invariant under symplectic transformations of the generators in $H_{-}$. Namely, under the transformation $\sigma \in {\rm Sp}(2g,\Z)$
\begin{equation*}
\sigma=\begin{pmatrix}
C_- & A_- \\
D_- & B_-
\end{pmatrix}\colon \
\begin{pmatrix}
b_- \\
a_-
\end{pmatrix}^{\sigma}= \sigma
\begin{pmatrix}
b_- \\
a_-
\end{pmatrix}
\end{equation*}
the potentials $\theta_{{\rm hom}}$ and $\theta_{G}$ remain the same, leaving the function $\mathcal{G}_B$ also invariant. The question how the change of Torelli marking affects the monodromy generating function was posed in~\cite{bertola2021wkb} and Proposition~\ref{Prop_3} allows us to provide the answer. Under the change \eqref{matrix_hom} of the canonical basis of $\mathcal{C}$ the Goldman potential $\theta_G$ remains invariant, while the homological potentials~$\theta^{\sigma}_{{\rm hom}}$,~$\theta_{{\rm hom}}$ for new and old Torelli markings are related by the term~\eqref{71}
\begin{equation*}\theta^{\sigma}_{{\rm hom}}=\theta_{{\rm hom}}+\delta\mathcal{G}_{{\rm hom}}. \end{equation*}
In our setting, with the help of \eqref{136} and \eqref{gen} one has
\begin{equation*}
Q_0=Q, \qquad Q_1=Q+6 \pi {\rm i} \sum_{1 \leq j \leq k \leq g}u_j u_k \frac{\partial}{\partial \Omega_{jk}}\log \det (C \Omega +D)
\end{equation*}
and
\begin{equation*}\delta\mathcal{G}_{{\rm hom}}=\delta\Bigg[\sum^n_{i=1}\pi {\rm i} r_i\Bigg({\rm reg}\int^{z^{(1)}_i}_{z^{(2)}_i}v_1 -{\rm reg}\int^{z^{(1)}_i}_{z^{(2)}_i}v_0 \Bigg)\Bigg] +6 \pi {\rm i} \delta \log \det (C \Omega +D),
\end{equation*}
where $v_0^2=Q_0$, $v_1^2=Q_1 $ define two different canonical coverings.
Combining that with the definition of the generating function~\eqref{yang_1}, we have

\begin{Proposition}
Under the change \eqref{matrix_hom} of the Torelli marking, the monodromy generating function transforms as
\begin{equation*}
\mathcal{G}^{\sigma}_B=\mathcal{G}_B+\sum^n_{i=1}\pi {\rm i} r_i\left({\rm reg}\int^{z^{(1)}_i}_{z^{(2)}_i}v_1 -{\rm reg}\int^{z^{(1)}_i}_{z^{(2)}_i}v_0 \right) +6 \pi {\rm i} \log \det (C \Omega +D),
\end{equation*}
where $v_0^2=Q$ and $v_1^2=Q+6 \pi {\rm i} \sum_{1 \leq j \leq k \leq g}u_j u_k \frac{\partial}{\partial \Omega_{jk}}\log \det (C \Omega +D)$.
\end{Proposition}

\section[Generalized WKB expansion of the monodromy generating function]{Generalized WKB expansion\\ of the monodromy generating function}\label{section4}

\subsection{WKB approximation of the Schr\"odinger equation}\label{section4.1}
To study the asymptotic expansion of the monodromy generating function $\mathcal{G}_B$, we consider the second-order equation in the form
\begin{equation}\label{140}
\partial^2\phi+\left(\frac{1}{2}S_B-\frac{Q_1}{\hbar}-\frac{Q}{\hbar^2}\right)\phi=0,
\end{equation}
where $Q \in \mathcal{Q}_{g,n}[\mathbf{r}]$, while $Q_1$ is a fixed meromorphic quadratic differential assumed to depend holomorphically on moduli of $\mathcal{M}_{g,n}$, with at most simple poles at the punctures $(z_j)^n_{j=1}$.

The WKB approximation for this equation is performed in the following way:
consider the canonical double cover $\hat{\mathcal{C}}_\hbar$ given by the equation
\begin{equation*}
 v_{\hbar}^2=\frac{Q}{\hbar^2}.
 \end{equation*}
Rescaling the differential $v=\hbar v_\hbar$ pass to the cover $v^2=Q$ which is now independent of $\hbar$. Choose some base point $x_0.$ In terms of local coordinate $z(x)=\int^x_{x_0}v$ and the function $\varphi(x)=\phi\sqrt{v(x)}$ equation \eqref{140} takes the form
\begin{equation}\label{142}
 \varphi_{zz}+\big(q(z)-p(z)\hbar^{-1}-\hbar^{-2}\big)\varphi=0,
\end{equation}
where
\begin{equation*}
q=\frac{S_B-S_v}{2v^2}, \qquad p=\frac{Q_1}{v^2},
\end{equation*}
(notice that in local coordinate $z(x)$ the Schwarzian projective connection \eqref{47} vanishes). Introducing the asymptotic series $s=\sum^{\infty}_{k=-1}\hbar^k s_k$, we write the solution for \eqref{142} in the form
\begin{equation*}
f_{x_0}=v^{-\frac{1}{2}}\exp \left(\int^x_{x_0}\big(\hbar^{-1}s_{-1}+s_0+\hbar s_1 +\cdots \big)v\right) ,
\end{equation*}
where $s_k$ are meromorphic functions on $\hat{\mathcal{C}}.$ The asymptotic series $s$ satisfies the Ricatti equation:%
\begin{equation}\label{145}
{\rm d}s+s^2v=-qv+\hbar^{-1}pv+\hbar^{-2}v.
\end{equation}
Then plugging its expansion into \eqref{145} and comparing terms near the same powers of $\hbar$ one gets the following first terms~$s_k$:
\begin{equation}\label{146}
s_{-1}=\pm 1,\qquad s_0=\frac{p}{2s_{-1}},\qquad s_1=-\frac{d p}{4v} -\frac{1}{2 s_{-1}}\left(\frac{p^2}{4}+q \right),
\end{equation}
while the consecutive terms satisfy the recurrence relation
\begin{equation}\label{147}s_{k+1}=-\frac{1}{2s_{-1}}\Bigg(\frac{{\rm d} s_k}{v}+\sum_{\substack {j+l=k \\ j,l\geq 0}}s_j s_l \Bigg), \qquad k \geq 1. \end{equation}
In particular, when $k=2$, we get
\begin{equation*}
s_2= \frac{1}{8v}{\rm d}\left(\frac{{\rm d}p}{s_{-1}v}+\frac{p^2}{2}+2q \right)-\frac{p {\rm d}p}{8v}+\frac{1}{4 s_{-1}}\left(qp-\frac{p^3}{4}\right).\end{equation*}

There is an ambiguity in choosing the value of $s_{-1}$, which corresponds to the choice of the sign for the square root $\sqrt{Q}$. Further below, we shall assume that $s_{-1}=+1$. To obtain another asymptotic series corresponding to $s_{-1}=-1$ it is sufficient to apply involution $\mu$ to get $\mu^*v=-v$. We define \textit{even} and \textit{odd} part of the asymptotic series~$s$ by
\begin{equation*}
s_{\rm odd}=\frac{1}{2}(s +\mu^*s ), \qquad s_{\rm even}=\frac{1}{2}(s -\mu^*s ).
\end{equation*}
Notice that $\mu^*s_{\rm odd}=s_{\rm odd}$ and $\mu^*s_{\rm even}=-s_{\rm even} $,
while
\begin{equation*}
\mu^*(s_{\rm odd}v)=-s_{\rm odd}v, \qquad \mu^*(s_{\rm even}v)=s_{\rm even}v.
\end{equation*}

\begin{Lemma}
The following equation holds:
\begin{equation*}
{\rm d} s_{\rm odd}=- 2s_{\rm even} s_{\rm odd}v .
\end{equation*}
\end{Lemma}
\begin{proof}
Expressing $s=s_{\rm odd}+s_{\rm even}$ and plugging it into~\eqref{145} we have
\begin{equation*}
{\rm d}(s_{\rm odd}+s_{\rm even})+\big(s^2_{\rm odd} + s^2_{\rm even} + 2s_{\rm odd}s_{\rm even}\big)v=-qv-\hbar^{-1}pv-\hbar^{-2}v.
\end{equation*}
This equality contains terms both symmetric and skew-symmetric under involution. Comparing only symmetric terms, one gets
\begin{equation*}
{\rm d} s_{\rm odd}+2 s_{\rm even} s_{\rm odd}v=0 .\tag*{\qed}
\end{equation*} \renewcommand{\qed}{}
\end{proof}

Using this relation, it is easy to obtain two WKB-solutions for the equation \eqref{140}:
\begin{equation}\label{154}f^{\pm}_{x_0}=\frac{1}{(s_{\rm odd}v)^{\frac{1}{2}}}\exp \left[\pm \int^{x}_{x_0}s_{\rm odd}v \right],\end{equation}
where $x_0$ is often chosen to be one of the branch points $x_i$, called turning point (see, e.g.,~\cite{kawai2005algebraic}). The differential $s_{\rm odd}v$ is multi-valued on the base curve $\mathcal{C}$ and generically singular at $x_i$. To define the integral correctly we pass to the double cover $\hat{\mathcal{C}}$ where $s_{\rm odd}v$ is well-defined. Skew-symmetry of $s_{\rm odd}v$ implies it has a vanishing residue at~$x_i$. Therefore, we can define the integral~by%
\begin{equation*}
\int^x_{x_i}s_{\rm odd}v=\frac{1}{2}\int^{x^{(1)}}_{x^{(2)}}s_{\rm odd}v,\end{equation*}
where we join preimages $x^{(1)}$ and $x^{(2)}$ of $x$ by an arc passing through the branch cut, which connects $x_i$ with some another branch point.

Introduce the meromorphic differentials
\begin{equation*}
v_k=(s_{\rm odd})_kv.
\end{equation*}
Analytic continuation of the WKB-solutions \eqref{154} along the edges of graph $\Sigma_Q$ gives rise to the relation between the homological shear coordinates and the Voros symbols -- integrals of $s_{\rm odd}v$ over the elements of~$H_{-}.$ The following proposition generalizes the one stated in \cite{bertola2021wkb} to the case when $Q_1 \neq 0$ and is proven in complete analogy.
\begin{Proposition}
For each $l \in H_-$ the homological shear coordinate $\rho_l$ admits the following asymptotic expansion
\begin{equation}\label{157}
 \rho_{l}(\hbar) \sim \int_l s_{\rm odd}v=\frac{1}{\hbar}\int_l v_{-1}+ \int_l v_0 + \hbar \int_l v_1 +\cdots, \qquad \hbar \xrightarrow[]{} 0^+,
 \end{equation}
 where the relation is understood modulo an addition of $ \pi{\rm i} k$, $k \in \Z.$
\end{Proposition}

\begin{Remark}The similar result was present in \cite{voros_alegr} where the meromorphic potential in~\eqref{1} was arranged in a different way, such that the double poles with fixed biresidues were attached to the reference meromorphic projective connection, while in our case these double poles belong to the quadratic differential~$Q$.
\end{Remark}

\subsection{WKB expansion of the Yang--Yang function}

The requirement for the monodromy map of the equation \eqref{140} to be a symplectomorphism imposes a restriction on the differential~$Q_1$. We can regard
\begin{equation*}
S=S_B-\frac{2Q_1}{\hbar}
\end{equation*} as a chosen base projective connection, then the condition for it to be admissible is ruled by Theorem~\ref{Th_5}. Namely, the form $\Theta_{(Q_1)}$, defined locally on $\mathcal{M}_{g,n}$, must be closed, $\delta\Theta_{(Q_1)}=0$.
The generating function $\mathcal{G}_B(\hbar)$ of this symplectomorphism is defined by
 \begin{equation} \label{yang_wkb}
 \delta \mathcal{G}_{B}(\hbar)=\mathcal{F}_{(S_B-{2Q_1}/{\hbar})}^*\theta_{G}(\hbar)-\theta_{{\rm hom}}(\hbar), \end{equation}
 where the symplectic potentials $\theta_G$ and $\theta_{{\rm hom}}$ are defined by \eqref{pot_hom} and \eqref{pot_G}:
\begin{equation} \label{pot_G_wkb}
 \theta_{G}(\hbar)=\sum^{g^-}_{j=1} \big(\rho_{b_j} \delta \rho_{a_j} - \rho_{a_j} \delta \rho_{b_j} \big)(\hbar),
 \end{equation}
 here $\rho_l(\hbar)$ is the homological shear coordinate corresponding to a loop $l \in H_{-}$ and
 \begin{equation}\label{pot_hom_wkb}\theta_{{\rm hom}}(\hbar)=\frac{1}{\hbar^2}\sum^{g^-}_{j=1} \big( B_j \delta A_j - A_j \delta B_j \big),\end{equation}
 here $\big(A_j=\oint_{a^-_j}v, B_j=\oint_{b^-_j}v\big)$ are period coordinates on $\mathcal{Q}_{g,n}[\mathbf{r}]$.
Using the pairing notation~\eqref{77} the symplectic potential $\theta_{{\rm hom}}$ in period coordinates reads as
\begin{equation*}
\theta_{{\rm hom}}(\hbar)=\frac{1}{\hbar^2}\left\langle \oint v,\delta \oint v\right\rangle .
\end{equation*}
The potential $\theta_G$
in homological shear coordinates $\rho_{l}(\hbar)$ by means of the expansion~\eqref{157} has the following expression:
\begin{equation}\label{160}
\theta_{G}(\hbar)=\sum^{\infty}_{i=-2}\hbar^i \sum_{\substack {l+k=i \\ l,k\geq -1}}\left\langle \oint v_l, \delta \oint v_k \right\rangle. \end{equation}
Meromorphic differentials $v_k$ could be obtained by antisymmetrizing the differentials $s_k v$, where functions $s_k$ are given by \eqref{146}, \eqref{147}. First four differentials $v_k$ take the following form:
\begin{gather} \label{163}
v_{-1}=v, \qquad v_0=\frac{Q_1}{2v},
\\
\label{164}
v_1=-\frac{Q^2_1}{8v^3}-\frac{qv}{2}, \qquad v_2=\frac{1}{4} \left(q\frac{Q_1}{v}-\frac{Q^3_1}{v^5}\right).
\end{gather}

By plugging \eqref{160} in \eqref{yang_wkb}, we see that the coefficient in front of $\hbar^{-2}$ in the expansion of~\eqref{yang_wkb} vanishes and
\begin{equation*}
\delta \mathcal{G}_B(\hbar) =\sum^{\infty}_{i=-1}\hbar^i \sum_{\substack {l+k=i \\ l,k\geq -1}}\left\langle \oint v_l, \delta \oint v_k \right\rangle := \sum^{\infty}_{i=-1} \hbar^i \delta G_i, \qquad \hbar \xrightarrow[]{} 0^+.
\end{equation*}

\subsubsection[Term G\_\{-1\}]{Term $\boldsymbol{G_{-1}}$}

Equation for $\delta G_{-1}$:
 \begin{equation*}
 \delta G_{-1}=\left\langle \oint v,\delta \oint v_0\right\rangle +\left\langle \oint v_0,\delta \oint v\right\rangle \end{equation*}
using the expressions \eqref{163}, \eqref{164} for $v_k$ could be written as follows:
 \begin{equation*}
 \delta G_{-1}=\frac{1}{2}\left\langle \oint v, \delta \oint \frac{Q_1}{v}\right\rangle +\frac{1}{2}\left\langle \oint \frac{Q_1}{v}, \delta \oint v\right\rangle.
 \end{equation*}
Notice that (after relabeling $\tilde{Q} \xrightarrow[]{} Q_1$, $v_0 \xrightarrow[]{} v$) this is exactly the term~\eqref{105} near $\hbar^1$ appearing in the expansion of the potential $\theta_1$ in the proof of Proposition~\ref{Prop_3}. Thus, we immediately get
\begin{equation*}
\delta G_{-1} =\Theta_{(Q_1)}+\delta\left[\sum^n_{j=1}\frac{\pi{\rm i} r_j}{2}
\int^{z^{(1)}_j}_{z^{(2)}_j}\frac{{Q}_1}{v} \right],
\end{equation*}
where before we assumed that the form $ \Theta_{(Q_1)}$ is closed on $\mathcal{M}_{g,n}$.
Then the integration leads to%
\begin{equation*}
G_{-1} =\Hat{G}_{(Q_1)}+\sum^n_{j=1}\frac{\pi{\rm i} r_j}{2}
\int^{z^{(1)}_j}_{z^{(2)}_j}\frac{{Q}_1}{v} ,
\end{equation*}
where there exists local holomorphic function $\Hat{G}_{(Q_1)}$ on the moduli space $\mathcal{M}_{g,n}$, such that
\begin{equation} \label{funG} \delta\Hat{G}_{(Q_1)}= \Theta_{(Q_1)}.\end{equation}
The geometrical meaning of the term $G_{-1}$ is that the condition of closedness of $\delta G_{-1}$ (or equivalently the existence of $\Hat{G}_{(Q_1)}$) is an obstruction for the monodromy map $\mathcal{F}_{(S_B-{Q_1}/{\hbar})}$ to be a~symplectomorphism.

\subsubsection[Term G\_0]{Term $\boldsymbol{G_0}$}
Equation for $\delta G_{0}$:
\begin{equation*}
\delta G_{0}=\left\langle \oint v_0,\delta \oint v_0\right\rangle +\left\langle \oint v,\delta \oint v_1\right\rangle +\left\langle \oint v_1,\delta \oint v\right\rangle \end{equation*}
using the expressions \eqref{163}, \eqref{164} for $v_k$ could be written as follows:
\begin{gather}
\delta G_{0}=\left[\frac{1}{4}\left\langle \oint \frac{Q_1}{v}, \delta \oint \frac{Q_1}{v}\right\rangle -\frac{1}{8}\left\langle \oint v,\delta \oint \frac{Q^2_1}{v^3}\right\rangle -\frac{1}{8}\left\langle \oint \frac{Q^2_1}{v^3},\delta \oint v\right\rangle \right] \nonumber\\
\hphantom{\delta G_{0}=}{}
 -\left[\frac{1}{2}\left\langle \oint v,\delta \oint qv\right\rangle +\frac{1}{2}\left\langle \oint qv,\delta \oint v\right\rangle \right].\label{172}
\end{gather}

The term in the first bracket is the coefficient \eqref{106} near $\hbar^2$ in the expansion of the potential~$\theta_1$ in the proof in Proposition \ref{Prop_3} and it equals
\begin{equation*}
\delta\left[\sum^n_{j=1}
\pi{\rm i} r_j \binom{\frac{1}{2}}{2}\int^{z^{(1)}_j}_{z^{(2)}_j}\frac{{Q^2_1}}{v^{3}} \right]. \end{equation*}
To compute the term in the second bracket, we notice that
\begin{equation*}
\frac{1}{2}\left\langle \oint v,\delta \oint qv\right\rangle +\frac{1}{2}\left\langle \oint qv,\delta \oint v\right\rangle =\frac{1}{2}\delta\left\langle \oint v, \oint qv\right\rangle +\left\langle \oint qv,\delta \oint v\right\rangle. \end{equation*}
It follows from \eqref{60} that $\big\langle \int qv,\delta \int v\big\rangle $ is a differential of the Bergman tau-function, namely
\begin{equation*}
\left\langle \oint qv,\delta \oint v\right\rangle =12 \pi{\rm i} \delta \log \tau_B |_{r}.\end{equation*}
Applying the variational formulas \eqref{55}, \eqref{56} and the homogeneity property \eqref{58} of the function $\tau_B$ on the full space $\mathcal{Q}_{g,n}$, the term $\big\langle \oint v, \oint qv\big\rangle $ could be written as
\begin{align*}
\left\langle \oint v,\oint qv\right\rangle& =-12 \pi{\rm i} \sum^{g^-}_{j=1}\left(A_j \frac{\delta}{\delta A_j}+B_j \frac{\delta}{\delta B_j}\right) \log \tau_B\\
& =-12 \pi{\rm i} \Bigg[\frac{5(2g-2+n)}{72}-\sum^n_{j=1}r_j \frac{\delta}{\delta r_j}\log \tau_B \Bigg] \\
& =-12 \pi{\rm i} \left[\frac{5(2g-2+n)}{72}-\sum^n_{j=1}\frac{r_j}{12} \int^{z^{(1)}_j}_{z^{(2)}_j}\left(qv+\frac{1}{4r^2_k}v \right) \right].
\end{align*}
Restricting this formula to $\mathcal{Q}_{g,n}[\mathbf{r}]$, we get
\begin{equation*}
\delta\left\langle \oint v, \oint qv\right\rangle =\delta\left[\sum^n_{j=1}{\pi{\rm i} r_j} \int^{z^{(1)}_j}_{z^{(2)}_j}\left(qv+\frac{1}{4r^2_k}v \right) \right],
\end{equation*}
(alternatively, this term could be computed via the resides after applying the Riemann bilinear identity to differentials $v$ and~$qv$). Putting all terms together in~\eqref{172} and integrating we obtain%
\begin{equation}\label{178}G_0=-12 \pi{\rm i} \log \tau_B|_{r} -\sum^n_{j=1}\frac{\pi{\rm i} r_j}{2} \int^{z^{(1)}_j}_{z^{(2)}_j}\left(qv+\frac{1}{4r^2_k}v \right)+\sum^n_{j=1}
\pi{\rm i} r_j \binom{\frac{1}{2}}{2}\int^{z^{(1)}_j}_{z^{(2)}_j}\frac{{Q^2_1}}{v^{3}}.
\end{equation}

\subsubsection[Term G\_1]{Term $\boldsymbol{G_1}$}
Equation for $\delta G_{1}$:
\begin{equation*}
\delta G_{1}=\left\langle \oint v,\delta \oint v_2\right\rangle +\left\langle \oint v_0,\delta \oint v_1\right\rangle +\left\langle \oint v_1,\delta \oint v_0\right\rangle +\left\langle \oint v_2,\delta \oint v\right\rangle
\end{equation*}
using the expressions \eqref{163}, \eqref{164} for $v_k$ could be written as follows:
\begin{gather*}
\delta G_{1}=\Bigg[{-}\frac{1}{16}\left\langle \oint \frac{Q^2_1}{v^3}, \delta \oint \frac{Q_1}{v}\right\rangle -\frac{1}{16}\left\langle \oint \frac{Q_1}{v}, \delta \oint \frac{Q^2_1}{v^3}\right\rangle \\
\hphantom{\delta G_{1}=}{}
+\frac{1}{16}\left\langle \oint v,\delta \oint \frac{Q^3_1}{v^5}\right\rangle +\frac{1}{16}\left\langle \oint \frac{Q^3_1}{v^5},\delta \oint v\right\rangle \Bigg]\\
\hphantom{\delta G_{1}=}{}
+\Bigg[{-}\frac{1}{4}\left\langle \oint \frac{Q_1}{v},\delta \oint {qv}\right\rangle -\frac{1}{4}\left\langle \oint {qv},\delta \oint \frac{Q_1}{v}\right\rangle \\
\hphantom{\delta G_{1}=}{}
+\frac{1}{4}\left\langle \oint v,\delta \oint q\frac{Q_1}{v}\right\rangle +\frac{1}{4}\left\langle \oint q\frac{Q_1}{v},\delta \oint v\right\rangle \Bigg].
\end{gather*}
The term in the first brackets is the coefficient \eqref{106} near $\hbar^3$ in the expansion of the potential~$\theta_1$ in the proof in Proposition~\ref{Prop_3} and it equals
\begin{equation*}
\delta\left[\sum^n_{j=1}\pi{\rm i} r_j\binom{\frac{1}{2}}{3}
 \int^{z^{(1)}_j}_{z^{(2)}_j}\frac{Q^3_1}{v^{5}} \right].
\end{equation*}
To treat the term in the second bracket first notice that it could be rewritten as
\begin{gather}
-\frac{1}{4}\delta\left\langle \oint \frac{Q_1}{v},\oint {qv}\right\rangle -\frac{1}{2}\left\langle \oint {qv},\delta \oint \frac{Q_1}{v}\right\rangle \nonumber\\
\qquad{} +\frac{1}{4}\delta\left\langle \oint v, \oint q\frac{Q_1}{v}\right\rangle +\frac{1}{2}\left\langle \oint q\frac{Q_1}{v},\delta \oint v\right\rangle .\label{simp}\end{gather}
Lemma~\ref{Lem_1} in the form~\eqref{79} implies that
\begin{equation*}
\left\langle \oint {qv},\delta \oint \frac{Q_1}{v}\right\rangle = \frac{1}{2}\int_{\partial \hat{\mathcal{C}}_0}\left( qv
\int^x_{p_0} \delta \frac{Q_1}{v}\right)+\left\langle \oint q\frac{Q_1}{v},\delta \oint v\right\rangle, \end{equation*}
so \eqref{simp} becomes
\begin{equation}\label{pairs}-\frac{1}{4}\delta\left\langle \oint\frac{Q_1}{v},\oint {qv}\right\rangle +\frac{1}{4}\delta\left\langle \oint v, \oint q\frac{Q_1}{v}\right\rangle -\frac{1}{4}\int_{\partial \hat{\mathcal{C}}_0}\left( qv
\int^x_{p_0} \delta \frac{Q_1}{v}\right).\end{equation}
Let us consider the last integral. While Abelian differential $\delta\frac{Q_1}{v}$ is holomorphic, $qv$ has residueless 4-order poles at the branch points $(x_i)$ and simple poles at the punctures $\{z^{(1)}_j, z^{(2)}_j \}^n_{j=1}$. So, the integral over the boundary reduces to the computation of residues:
\begin{gather}\label{182}\int_{\partial \hat{\mathcal{C}}_0}\Bigg( qv
\int^x_{p_0} \delta \frac{Q_1}{v}\Bigg)=\sum^{4g-4+2n}_{i=1}2\pi{\rm i} \ \underset{x_i}{\rm res}\left[qv \int^x_{p_0} \delta\frac{Q_1}{v}\right]+\sum^{n}_{j=1}2\pi{\rm i} \underset{\{z^{(1)}_j, z^{(2)}_j \}}{\rm res}\left[qv \int^x_{p_0} \delta \frac{Q_1}{v}\right].\!\!\!\end{gather}
To compute residues near simple poles, we recall the formulas \eqref{q_form} for $q(x)$ and $S_v$ and use the local coordinate $\zeta$~\eqref{39} to write near ${z}^{(1)}_j$:
\begin{equation*}
 qv=\frac{S_B-S_v}{2v}=\frac{S_B(\zeta)-\frac{1}{2\zeta^2}}{2\frac{{r}_j}{\zeta}}{\rm d}\zeta=\left(-\frac{1}{4{r}_j \zeta} +O(1)\right) {\rm d}\zeta.
\end{equation*}
Due to skew-symmetry of $qv$ the expansion near ${z}^{(2)}_j$ is the negation of the above formula.
Thus,%
\begin{equation*}
\big(\underset{z^{(1)}_j}{\rm res}+\underset{z^{(2)}_j}{\rm res} \big)\left(qv \int^x_{p_0}\delta\frac{Q_1}{v}\right)=\delta \left[-\frac{1}{4{r}_j}\int^{{z}^{(1)}_j}_{{z}^{(2)}_j}\frac{Q_1}{v}\right],\end{equation*}
since the differential $\delta$ commutes with the line integral. To simplify the residue near a branch point $x_i$ at first notice that the variational formula \eqref{54} implies that the differential $\delta (qv)$ is holomorphic at $x_i,$ so
\begin{equation*}
\underset{x_i}{\rm res}\left[qv \int^x_{p_0} \delta\frac{Q_1}{v}\right]
=\delta\left(\underset{x_i}{\rm res}\left[qv \int^x_{p_0} \frac{Q_1}{v}\right]\right). \end{equation*}
Let’s assume that the Bergman projective connection admits the following expansion near $x_i$ on the base curve $\mathcal{C}$ in local coordinate \eqref{38}:
\begin{equation*}S_B(\xi)=S_B (x_i)+{S_B}' (x_i) \xi +\cdots. \end{equation*}
Being lifted to $\hat{\mathcal{C}}$, it transforms like
\begin{equation*}S_B(\hat{\xi})\big({\rm d}\hat{\xi}\big)^2=S_B({\xi})({\rm d}{\xi})^2+S\big({\xi}, \hat{\xi}\big), \end{equation*}
where $\hat{\xi}$ is local coordinate \eqref{37} near $x_i$ on $\hat{\mathcal{C}}$, $S\big({\xi}, \hat{\xi}\big)$ is the Schwarzian derivative
\begin{equation*}
S\big({\xi}, \hat{\xi}\big)=\left(\frac{{\xi}''}{{\xi}'} \right)'-\frac{1}{2}\left(\frac{{\xi}''}{{\xi}'} \right)^2, \end{equation*}
where derivatives are taken with respect to $\hat{\xi}$. Having that ${\xi}=\hat{\xi}^2$, we write
\begin{equation*}
S_B\big(\hat{\xi}\big)\big({\rm d}\hat{\xi}\big)^2=4(S_B (x_i)+{S_B}' (x_i) \hat{\xi}^2)\hat{\xi}^2 \big({\rm d}\hat{\xi} \big)^2-\frac{3}{2\hat{\xi}^2}\big({\rm d}\hat{\xi} \big)^2.\end{equation*}
Also
\begin{equation*}
S_v=-\frac{4}{\hat{\xi}^2}\big({\rm d}\hat{\xi}\big)^2, \qquad v=3\hat{\xi}^2 {\rm d}\hat{\xi},\end{equation*}
leading to
\begin{equation*}
qv=\left[\frac{5}{12\hat{\xi}^4}+O(1) \right] {\rm d}\hat{\xi}.
\end{equation*}
Then the residue near $x_i$ could be expressed as
\begin{equation*}
\underset{x_i}{\rm res}\left[qv \int^x_{p_0} \frac{Q_1}{v}\right]=\underset{x_i}{\rm res} \left(\frac{5d \hat{\xi}}{12\hat{\xi}^4} \int^x_{p_0}\frac{Q_1}{v}\right)=\frac{5}{12}\frac{1}{3!}\left( \frac{(Q_1/v)}{ d\hat{\xi}} \right)''(x_i)=\frac{5}{36}\underset{x_i}{\rm res}\left(\frac{Q_1/v}{\int^x_{x_i}v} \right).
\end{equation*}
The integral \eqref{182} takes the following form:
\begin{equation}\label{pt1}\int_{\partial \hat{\mathcal{C}}_0}\left( qv
\int^x_{p_0} \delta \frac{Q_1}{v}\right)=\delta \left[\sum^{4g-4+2n}_{i=1} \frac{5 \pi{\rm i}}{18}\underset{x_i}{\rm res}\left(\frac{Q_1/v}{\int^x_{x_i}v} \right)-\sum^{n}_{j=1}\frac{\pi{\rm i}}{2{r}_j}\int^{{z}^{(1)}_j}_{{z}^{(2)}_j}\frac{Q_1}{v}\right].\end{equation}
The differential of the first pairing $\delta\big\langle \oint \frac{Q_1}{v}, \oint {qv}\big\rangle $ in \eqref{pairs} is computed by analogy. Applying the Riemann bilinear identity, one has
\begin{equation*}\left\langle \oint \frac{Q_1}{v},\oint {qv}\right\rangle =-\frac{1}{2}\int_{\partial \hat{\mathcal{C}}_0}\left( qv
\int^x_{p_0} \frac{Q_1}{v}\right), \end{equation*}
resulting in
\begin{equation}\label{pt2}
\delta\left\langle \oint \frac{Q_1}{v},\oint {qv}\right\rangle =\delta \left[-\sum^{4g-4+2n}_{i=1} \frac{5 \pi{\rm i}}{36}\underset{x_i}{\rm res}\left(\frac{Q_1/v}{\int^x_{x_i}v} \right)+\sum^{n}_{j=1}\frac{\pi{\rm i}}{4{r}_j}\int^{{z}^{(1)}_j}_{{z}^{(2)}_j}\frac{Q_1}{v}\right].
\end{equation}
Finally, applying the Riemann bilinear identity to the pairing $\big\langle \oint v, \oint q\frac{Q_1}{v}\big\rangle $, we have
\begin{equation*}
\left\langle \oint v, \oint q\frac{Q_1}{v}\right\rangle =-\sum^{4g-4+2n}_{i=1}\pi{\rm i} \ \underset{x_i}{\rm res}\left[q \frac{Q_1}{v} \int^x_{p_0} v\right]+\sum^{n}_{j=1}\pi{\rm i} \ \underset{\{z^{(1)}_j, z^{(2)}_j \}}{\rm res}\left[v \int^x_{p_0} q\frac{Q_1}{v}\right].\end{equation*}

While
\begin{equation*}
\underset{\{z^{(1)}_j, z^{(2)}_j \}}{\rm res}\left[v \int^x_{p_0} q\frac{Q_1}{v}\right]=r_j\int^{z^{(1)}_j}_{z^{(2)}_j} q\frac{Q_1}{v},\end{equation*}
the residue near $x_i$ is computed noticing that
\begin{equation*}
q \frac{Q_1}{v}=\left(\frac{5}{36\hat{\xi}^6}+O\left(\frac{1}{\hat{\xi}^2}\right) \right)\frac{Q_1}{v}, \qquad \int^x_{p_0}v=C(p_0)+{\hat{\xi}^3} .\end{equation*}
Differential $\frac{Q_!}{v}$ is antisymmetric, and, thus, expands by even powers of $\hat{\xi}$ near a branch point. That leads to
\begin{equation*}
\underset{x_i}{\rm res}\left[q \frac{Q_1}{v} \int^x_{p_0} v\right]=\frac{5}{36}\underset{x_i}{\rm res}\left[\frac{Q_1/v}{\hat{\xi}^3} \right]=\frac{5}{36}\underset{x_i}{\rm res}\left(\frac{Q_1/v}{\int^x_{x_i}v} \right).\end{equation*}
Overall,
\begin{equation}\label{pt3}\delta\left\langle \oint v, \oint q\frac{Q_1}{v}\right\rangle =\delta \left[-\sum^{4g-4+2n}_{i=1} \frac{5 \pi{\rm i}}{36}\underset{x_i}{\rm res}\left(\frac{Q_1/v}{\int^x_{x_i}v} \right)+\sum^{n}_{j=1}\pi{\rm i} r_j\int^{z^{(1)}_j}_{z^{(2)}_j} q\frac{Q_1}{v}\right].\end{equation}

Plugging derived expressions \eqref{pt1}, \eqref{pt2} and \eqref{pt3} into \eqref{pairs} and integrating, we obtain
\begin{gather}G_{1}=-\sum^{4g-4+2n}_{i=1} \frac{5 \pi{\rm i}}{72}\underset{x_i}{\rm res}\left(\frac{Q_1/v}{\int^x_{x_i}v} \right)+\sum^{n}_{j=1}\frac{\pi{\rm i} r_j }{4}\int^{z^{(1)}_j}_{z^{(2)}_j} q\frac{Q_1}{v} \nonumber\\
\hphantom{G_{1}=}{}
+\sum^{n}_{j=1}\frac{\pi{\rm i}}{16{r}_j}\int^{{z}^{(1)}_j}_{{z}^{(2)}_j}\frac{Q_1}{v}+\sum^n_{j=1}\pi{\rm i} r_j\binom{\frac{1}{2}}{3}
 \int^{z^{(1)}_j}_{z^{(2)}_j}\frac{Q^3_1}{v^{5}}.\label{trm1}
 \end{gather}
To sum up, we can formulate the following theorem
\begin{Theorem}
Consider the differential equation
\begin{equation*}
\partial^2\phi+\left(\frac{1}{2}S_B-\frac{Q_1}{\hbar}-\frac{Q}{\hbar^2}\right)\phi=0
\end{equation*}
on a Riemann surface $\mathcal{C}$. Let $Q$ be a GMN quadratic differential on $\mathcal{C}$ with simple zeroes and~$n$ second-order poles at $z_1,\dots, z_n$ and biresidues $r^2_1,\dots,r^2_n$. $Q_1$~is a meromorphic quadratic differential which depends holomorphically on moduli of $\mathcal{M}_{g,n}$, with at most simple poles at~$z_j$. $S_B$~is the Bergman projective connection~\eqref{12} defined with respect to some Torelli marking on~$\mathcal{C}$. Differential $Q$ defines a canonical double cover $\hat{\mathcal{C}}$ via $v^2=Q$ and gives rise to an ideal triangulation of $\mathcal{C}$ used in the definition of the homological shear coordinates \eqref{homsc}. For the chosen base projective connection $S_B-\frac{2Q_1}{\hbar}$ denote by $\mathcal{F}_{(S_B-2Q_1/\hbar)}$ the monodromy map between the moduli space $\mathcal{Q}_{g,n}[\mathbf{r}/\hbar]$ of pairs $\big(\mathcal{C}, Q/\hbar^2\big)$ and the symplectic leaf ${\rm CV}_{g,n}[\mathbf{m}(\hbar)]$ of the ${\rm PSL}(2, \CC)$ character variety, where each $m_j(\hbar)$ is related to $r_j$ via $\eqref{7}, \eqref{8}$ as
\begin{equation*}
\frac{r_j}{\hbar}=\pm \left[\frac{\log m_j}{2 \pi{\rm i} } \left(\frac{\log m_j}{2 \pi{\rm i} }-1 \right)\right]^{1/2}.
\end{equation*}

The map $\mathcal{F}_{(S_B-2Q_1/\hbar)}$ is a symplectomorphism, provided that the $1$-form $\Theta_{(Q_1)}$, locally defined on $\mathcal{M}_{g,n}$, is closed, $\delta\Theta_{(Q_1)}=0$.

Introduce the symplectic potential $\theta_{{\rm hom}}$ \eqref{pot_hom_wkb} of the homological symplectic form on $\mathcal{Q}_{g,n}[\mathbf{r}/\hbar]$ and symplectic potential $\theta_G$ \eqref{pot_G_wkb} for the Goldman form on ${\rm CV}_{g,n}[\mathbf{m}(\hbar)]$
The generating function~$\mathcal{G}_{B}$ of the monodromy symplectomorphism between $\mathcal{Q}_{g,n}[\mathbf{r}/\hbar]$ and ${\rm CV}_{g,n}[\mathbf{m}(\hbar)]$ is defined~by%
\begin{equation} \label{yang}
\delta \mathcal{G}_{B}(\hbar)=\mathcal{F}_{(S_B-{2Q_1}/{\hbar})}^*\theta_{G}(\hbar)-\theta_{{\rm hom}}(\hbar) \end{equation}
and has the following asymptotics
 as $\hbar \xrightarrow{} 0^+$:
\begin{equation*}\mathcal{G}_B(\hbar)=\frac{G_{-1}}{\hbar}+G_0+G_1 \hbar + O\big(\hbar^2\big).\end{equation*}
Here
\begin{equation*}G_{-1} =\Hat{G}_{(Q_1)}+\sum^n_{j=1}\frac{\pi{\rm i} r_j}{2}
\int^{z^{(1)}_j}_{z^{(2)}_j}\frac{{Q}_1}{v} , \end{equation*}
where the function $\Hat{G}_{(Q_1)}$ is defined by \eqref{funG}. Its explicit form depends on the concrete choice of~$Q_1$;
\begin{equation*}
G_{0}=-12 \pi{\rm i} \log \tau_B|_{r} -\sum^n_{j=1}\frac{\pi{\rm i} r_j}{2} \int^{z^{(1)}_j}_{z^{(2)}_j}\left(qv+\frac{1}{4r^2_j}v \right)+\sum^n_{j=1}
\pi{\rm i} r_j \binom{\frac{1}{2}}{2}\int^{z^{(1)}_j}_{z^{(2)}_j}\frac{{Q^2_1}}{v^{3}},
\end{equation*}
here $\log \tau_B|_{r}$ is the Bergman tau-function defined by \eqref{60} on stratum of the moduli space of quadratic differentials with second-order poles, $q(x)$ is a meromorphic function on $\mathcal{C}$ given by
\begin{equation*}q(x)=\frac{S_B-S_v}{2v^2}, \end{equation*}
here $S_v$ is the Schwarzian projective connection \eqref{q_form};
\begin{gather*}
G_1=-\sum^{4g-4+2n}_{i=1} \frac{5 \pi{\rm i}}{72}\underset{x_i}{\rm res}\left(\frac{Q_1/v}{\int^x_{x_i}v} \right)+\sum^{n}_{j=1}\frac{\pi{\rm i} r_j }{4}\int^{z^{(1)}_j}_{z^{(2)}_j} q\frac{Q_1}{v} \\
\hphantom{G_1=}{}
+\sum^{n}_{j=1}\frac{\pi{\rm i}}{16{r}_j}\int^{{z}^{(1)}_j}_{{z}^{(2)}_j}\frac{Q_1}{v}+\sum^n_{j=1}\pi{\rm i} r_j\binom{\frac{1}{2}}{3}
 \int^{z^{(1)}_j}_{z^{(2)}_j}\frac{Q^3_1}{v^{5}},
 \end{gather*}
where the first sum runs over the branch points of the double cover.
\end{Theorem}

We can propose an alternative way of computing the expansion of Yang--Yang function $\mathcal{G}_{B}(\hbar)$, assuming we obtained the expansion of the related generating function $\mathcal{G}^0_{B}(\hbar)$ for $Q_1 \equiv 0.$ Rewrite the equation \eqref{140} in the form
\begin{equation*}\partial^2\phi+\left(\frac{1}{2}S_B-\frac{Q+ \hbar Q_1}{\hbar^2}\right)\phi=0,\end{equation*}
corresponding to detachment of differential $\frac{Q_1}{\hbar}$ from projective connection $S_B$ and joining it with~$\frac{Q_!}{\hbar^2}$, so that $\frac{Q+ \hbar Q_1}{\hbar^2} \in \mathcal{Q}_{g,n}[\mathbf{r}/\hbar]$. This operation induces the following diagram of maps, where symplectomorphism $\mathcal{F}_{(S_B-{2Q_1}/{\hbar})}$ factors through the composition $\mathcal{F}_{(S_B)} \circ \mathcal{H} $ of symplectomorphisms:
\[
 \begin{tikzcd}
{\mathcal{Q}_{g,n}[\mathbf{r}/\hbar]} \arrow[rr, " Q \xrightarrow{\mathcal{H}} Q+\hbar Q_1"] \arrow[rrd, "\mathcal{F}_{(S_B-{2Q_1}/{\hbar})}"'] & & {\mathcal{Q}_{g,n}[\mathbf{r}/\hbar] \ar[d, "\mathcal{F}_{(S_B)}"]} \\
 & & {{\rm CV}_{g,n}[\mathbf{m}(\hbar)].}
\end{tikzcd}
\]
Denote by $\theta^{1}_{{\rm hom}}(\hbar)$ the symplectic potential computed
via the periods of the Abelian differential~$v^\hbar_1$, which defines canonical double cover $\hat{\mathcal{C}_\hbar}$ by $\big(v^\hbar_{1}\big)^2=Q+\hbar Q_1$. Then the Monodromy generating function \eqref{yang} admits the following representation:
\begin{equation*} \delta \mathcal{G}_{B}(\hbar)=\mathcal{H}^*\big(\mathcal{F}_{(S_B)}^*\theta_{G}(\hbar)-\theta^{1}_{{\rm hom}}(\hbar)\big)+\big(\mathcal{H}^*\theta^{1}_{{\rm hom}}(\hbar)-\theta_{{\rm hom}}(\hbar)\big). \end{equation*}
The $\hbar$-expansion of the generating function for map $\mathcal{H}$ was computed in \eqref{126} along with the proof of Proposition \ref{Prop_3}:
\begin{equation*}\mathcal{H}^*\theta^{1}_{{\rm hom}}(\hbar)-\theta_{{\rm hom}}(\hbar)=\delta\left[\sum^n_{j=1}
\pi{\rm i} r_j \left[\sum^{\infty}_{k=1}\hbar^{k-2}\binom{\frac{1}{2}}{k}\int^{z^{(1)}_j}_{z^{(2)}_j}\frac{{Q^k_1}}{v^{2k-1}} \right]\right] + \hbar^{-1} \delta\Hat{G}_{(Q_1)}. \end{equation*}

Homological coordinates on $\mathcal{Q}_{g,n}[\mathbf{r}/\hbar]$, defined by the periods of $v^\hbar_{1}$, holomorphically depend on $\hbar.$ That allows us to perform the $\hbar$-expansion of the
\begin{equation}\label{exp}\mathcal{H}^*\big(\mathcal{F}_{(S_B)}^*\theta_{G}(\hbar)-\theta^{1}_{{\rm hom}}(\hbar)\big) \end{equation}
in two steps. At first, obtain the WKB-expansion of the generating function
\begin{equation*}
 \mathcal{F}_{(S_B)}^*\theta_{G}(\hbar)-\theta^{1}_{{\rm hom}}(\hbar),
\end{equation*}
assuming that $\hbar$ of $v^\hbar_{1}$ is fixed. This is equivalent to setting $Q_1=0$ and performing the computations as in Section~\ref{section4.1}. This expansion was considered in \cite{bertola2021wkb} and includes only even powers of~$\hbar$. We denote it by
\begin{equation*}\mathcal{G}^0_B(\hbar) =\sum^{\infty}_{i=0} \hbar^{2i} G^0_{2i}. \end{equation*}
Now vary $Q$ by the differential $\hbar Q_1$ and expand the terms $G^0_{2i}$ by Taylor series to obtain the full $\hbar$-expansion for~\eqref{exp}. The following proposition holds

\begin{Proposition}
Let \begin{equation} \label{G0}
\mathcal{G}^0_B(\hbar) =\sum^{\infty}_{i=0} \hbar^{2i} G^0_{2i}, \qquad \hbar \xrightarrow[]{} 0^+
\end{equation} be a WKB expansion of the Monodromy generating function $\delta\mathcal{G}^0_B(\hbar)=\mathcal{F}_{(S_B)}^*\theta_{G}(\hbar)-\theta_{{\rm hom}}(\hbar)$ of the equation
\begin{equation*}\partial^2\phi+\left(\frac{1}{2}S_B-\frac{Q }{\hbar^2}\right)\phi=0.\end{equation*}
Then the generalized WKB expansion
\begin{equation*}\mathcal{G}_B(\hbar) =\sum^{\infty}_{i=-1} \hbar^i G_i, \qquad \hbar \xrightarrow[]{} 0^+ \end{equation*}
of the monodromy generating function $\delta \mathcal{G}_{B}(\hbar)=\mathcal{F}_{(S_B-{2Q_1}/{\hbar})}^*\theta_{G}(\hbar)-\theta_{{\rm hom}}(\hbar)$ of the equation
\begin{equation*}\partial^2 \phi+\left(\frac{1}{2}S_B-\frac{{Q}_1}{\hbar}-\frac{Q}{\hbar^2}\right)\phi=0\end{equation*}
is related to \eqref{G0} by
\begin{gather}
G_{-1} =\Hat{G}_{(Q_1)}+\sum^n_{j=1}\frac{\pi{\rm i} r_j}{2}
\int^{z^{(1)}_j}_{z^{(2)}_j}\frac{{Q}_1}{v} ,\nonumber
\\
 \label{G_0 term}
G_{2k}=\sum^k_{i=0}\delta^{(2i)}_{Q_1}\frac{G^0_{2k-2i}}{(2i)!}+\sum^n_{j=1}\binom{\frac{1}{2}}{2k+2}
\pi{\rm i} r_j \int^{z^{(1)}_j}_{z^{(2)}_j}\frac{{Q}^{2k+2}_1}{v^{4k+3}}, \qquad k=0,\dots, \infty ,\\
\label{tpl}G_{2k+1}=\sum^k_{i=0}\delta^{(2i+1)}_{Q_1}\frac{G^0_{2k-2i}}{(2i+1)!}+\sum^n_{j=1}\binom{\frac{1}{2}}{2k+3}\pi{\rm i} r_j \int^{z^{(1)}_j}_{z^{(2)}_j}\frac{{Q}^{2k+3}_1}{v^{4k+5}}, \qquad k=0,\dots, \infty,
\end{gather}
where $\delta^{(k)}_{Q_1} f$ denotes $\frac{\partial^k}{\partial \hbar^k}f[Q+\hbar Q_1]\big|_{\hbar=0}$.
\end{Proposition}

\begin{Remark}
Comparing the derived expression \eqref{178} for $G_0$ with the formula \eqref{G_0 term}, we obtain
\[
G^0_0=-12 \pi{\rm i} \log \tau_B|_{r} -\sum^n_{j=1}\frac{\pi{\rm i} r_j}{2} \int^{z^{(1)}_j}_{z^{(2)}_j}\left(qv+\frac{1}{4r^2_k}v \right),
\]
which is exactly the leading term of the WKB expansion, previously computed in \cite{bertola2021wkb}.
\end{Remark}

We expect that the above result may be useful for relating the WKB expansion with the framework of topological recursion \cite{Eynard_2007}. Indeed, the formulas \eqref{G_0 term}, \eqref{tpl} involve variations of the spectral cover for a fixed base curve which appear in the recursive definition of Eynard--Orantin invariants. For a higher genus such variations were studied in detail in \cite{Baraglia_2019, bertola2019spaces, klimov2021var}. In particular, due to the relation \eqref{tpl} one may alternatively derive the term $G_1$ \eqref{trm1} by varying the Bergman tau-function appearing in $G^0_0$ using variational techniques of \cite{klimov2021var}.

\subsection*{Acknowledgements} The author thanks his scientific advisor D.~Korotkin for posing the problem and fruitful discussions, and is grateful to anonymous referees for carefully reading the manuscript and giving a~relevant contribution to enhance the paper.

\pdfbookmark[1]{References}{ref}
\LastPageEnding

\end{document}